\documentclass[twocolumn,numbook,draft]{svjour3}          
\smartqed  
\usepackage{graphicx}
\usepackage[utf8]{inputenc}
\usepackage{ltl}
\usepackage[color=teal!50]{todonotes}
\usepackage{amssymb}
\usepackage{tikz}
\usetikzlibrary{arrows,arrows.meta,automata,shapes,shapes.multipart,positioning,calc}
\usepackage[ruled,vlined,linesnumbered]{algorithm2e}
\usepackage{mathtools}
\usepackage{subcaption}
\usepackage{array}
\usepackage{hyperref}
\usepackage{pgfplots}
\pgfplotsset{compat=1.8}
\usepackage[capitalise,nameinlink]{cleveref}
\usepackage[stable]{footmisc}

\usepackage{cite}

\usepackage[capitalise,nameinlink]{cleveref}

\newcommand{\ie}{i.e.\@\xspace}

\newcommand{\eg}{e.g.\@\xspace}

\newcommand{\true}{\mathit{true}}
\newcommand{\false}{\mathit{false}}

\newcommand{\inputs}{I}
\newcommand{\outputs}{O}
\newcommand{\inp}[1]{\boldsymbol{#1 }}

\newcommand{\compatibleWords}[1]{\mathcal{C}(#1)}

\newcommand{\restrict}[2]{#1 \cap #2}

\newcommand{\project}[2]{#1_{\pi(#2)}}

\newcommand{\propositions}[1]{\textit{prop}(#1)}

\newcommand{\depGraph}[1]{\mathcal{D}_{#1}}

\newcommand{\Next}{\LTLnext}
\newcommand{\Globally}{\LTLsquare}
\newcommand{\Eventually}{\LTLdiamond}
\newcommand{\Until}{\LTLuntil}

\DeclareMathOperator{\pc}{||}

\newcommand{\var}{\texttt}

\newcommand{\realInp}[2]{\mu}

\begin{document}

\title{Specification Decomposition for Reactive Synthesis\thanks{This work was partially supported by the German Research Foundation~(DFG) as part of the Collaborative Research Center ``Foundations of Perspicuous Software Systems'' (TRR 248 -- CPEC, 389792660), and by the European Research Council (ERC) Grant OSARES (No. 683300). The authors thank Alexandre Duret-Lutz for providing valuable feedback on the algorithm and for bringing up the idea of extending assumption dropping to non-strict formulas. Moreover, they thank Marvin Stenger for help with the implementation.}
}


\author{Bernd Finkbeiner \and
		Gideon Geier \and
        Noemi Passing
}


\institute{B.~Finkbeiner and N.~Passing \at
              CISPA Helmholtz Center for Information Security, Germany \\
              \email{finkbeiner@cispa.de, noemi.passing@cispa.de}
           \and
           G.~Geier \at
              Saarland University, Germany \\
              \email{geier@react.uni-saarland.de}
}

\date{Received: date / Accepted: date}

\maketitle

\begin{abstract}
Reactive synthesis is the task of automatically deriving a correct implementation from a specification. It is a promising technique for the development of verified programs and hardware.
	Despite recent advances in terms of algorithms and tools, however, reactive synthesis is still not practical when the specified systems reach a certain bound in size and complexity.
	In this paper, we present a sound and complete modular synthesis algorithm that automatically decomposes the specification into smaller subspecifications. For them, independent synthesis tasks are performed, significantly reducing the complexity of the individual tasks.
	Our decomposition algorithm guarantees that the subspecifications are independent in the sense that completely separate synthesis tasks can be performed for them. Moreover, the composition of the resulting implementations is guaranteed to satisfy the original specification.
	Our algorithm is a preprocessing technique that can be applied to a wide range of synthesis tools. 
	We evaluate our approach with state-of-the-art synthesis tools on established benchmarks: The runtime decreases significantly when synthesizing implementations modularly.
\end{abstract}

\section{Introduction}

Reactive synthesis automatically derives an implementation that satisfies a given specification. It is a push-button method producing implementations which are correct by construction. Therefore, reactive synthesis is a promising technique for the development of probably correct systems since it allows for concentrating on \emph{what} a system should do instead of \emph{how} it should be done.

Despite recent advances in terms of efficient algorithms and tools, however, reactive synthesis is still not practical when the specified systems reach a certain bound in size and complexity. 
It is long known that the scalability of model checking algorithms can be improved significantly by using compositional approaches, \ie, by breaking down the analysis of a system into several smaller subtasks.~\cite{Compos97,ClarkeLM89}.
In this paper, we apply compositional concepts to reactive synthesis:
We present and extend a modular synthesis algorithm~\cite{FinalVersion} that decomposes a specification into several subspecifications. Then, independent synthesis tasks are performed for them. The implementations obtained from the subtasks are combined into an implementation for the initial specification. 
The algorithm uses synthesis as a black box and can thus be applied to a wide range of synthesis algorithms.
In particular, it can be seen as a preprocessing step for reactive synthesis that enables compositionality for existing algorithms and tools.

Soundness and completeness of modular synthesis strongly depends on the decomposition of the specification into subspecifications.
We introduce a criterion, \emph{non-contradictory independent sublanguages}, for subspecifications that ensures soundness and completeness: The original specification is equirealizable to the subspecifications and the parallel composition of the implementations for the subspecifications is guaranteed to satisfy the original specification. 
The key question is now how to decompose a specification such that the resulting subspecifications satisfy the criterion.

Lifting the language-based criterion to the automaton level, we present a decomposition algorithm for nondeterministic Büchi automata that directly implements the independent sublanguages paradigm. 
Thus, using subspecifications obtained with this decomposition algorithm ensures soundness and completeness of modular synthesis.
A specification given in the standard temporal logic LTL can be translated into an equivalent nondeterministic Büchi automaton and hence the decomposition algorithm can be applied as well.

However, while the decomposition algorithm is semantically precise, it utilizes several expensive automaton operations. For large specifications, the decomposition thus becomes infeasible.
Therefore, we present an approximate decomposition algorithm for LTL specification that still ensures soundness and completeness of modular synthesis but is more scalable.
It is approximate in the sense that, in contrast to the automaton decomposition algorithm, it does not necessarily find all possible decompositions.
Moreover, we present an optimization of the LTL decomposition algorithm for formulas in a common assumption-guarantee format. It analyzes the assumptions and drops those that do not influence the realizability of the rest of the formula, yielding more fine-grained decompositions.
We extend the optimization from specifications in a strict assume-guarantee format to specifications consisting of several conjuncts in assume-guarantee format. 
This allows for applying the optimization to even more of the common LTL synthesis benchmarks.

We have implemented both decomposition procedures as well as the modular synthesis algorithm and used it with the two state-of-the-art synthesis tools BoSy~\cite{BoSy} and Strix~\cite{MeyerStrix}. We evaluate our algorithms on the established benchmarks from the synthesis competition SYNTCOMP~\cite{SYNTCOMP}.
As expected, the decomposition algorithm for nondeterministic Büchi automata becomes infeasible when the specifications grow.
For the LTL decomposition algorithm, however, the experimental results are excellent: Decomposition terminates in less than 26 milliseconds on all benchmarks. Hence, the overhead of LTL decomposition is negligible, even for non-decomposable specifications.
Out of 39 decomposable specifications, BoSy and Strix increase their number of synthesized benchmarks by nine and five, respectively. For instance, on the generalized buffer benchmark~\cite{JacobsB16,Jobstmann07} with three receivers, BoSy is able to synthesize a solution within~28 seconds using modular synthesis while neither the non-compositional version of BoSy, nor the non-compositional version of Strix terminates within one hour.
For twelve and nine further benchmarks, respectively, BoSy and Strix reduce their synthesis times significantly, often by an order of magnitude or more, when using modular synthesis instead of their classical algorithms.
The remaining benchmarks are too small and too simple for compositional methods to pay off. 
Thus, decomposing the specification into smaller subspecifications indeed increases the scalability of synthesis on larger systems.

\textbf{Related Work:} 
Compositional approaches are long known to improve the scalability of model checking algorithms significantly~\cite{Compos97,ClarkeLM89}.
The approach that is most related to our contribution is a preprocessing algorithm for compositional model checking~\cite{DurejaR18}.
It analyzes dependencies between the properties that need to be checked in order to reduce the number of model checking tasks. We lift this idea from model checking to reactive synthesis. The dependency analysis in our algorithm, however, differs inherently from the one for model checking.

There exist several compositional approaches for reactive synthesis. 
The algorithm by Filiot~et~al. depends, like our LTL decomposition approach, heavily on dropping assumptions~\cite{FiliotJR10}. They use an heuristic that, in contrast to our criterion, is incomplete. While their approach is more scalable than a non-compositional one, one does not see as significant differences as for our algorithm. 
The algorithm by Kupferman~et~al. is designed for incrementally adding requirements to a specification during system design~\cite{KupfermanPV06}. Thus, it does not perform independent synthesis tasks but only reuses parts of the already existing solutions.
In contrast to our algorithm, both \cite{KupfermanPV06} and \cite{FiliotJR10} do not consider dependencies between the components to obtain prior knowledge about the presence or absence of conflicts in the implementations.

Assume-guarantee synthesis algorithms~\cite{ChatterjeeH07,MajumdarMSZ20,FinkbeinerP21,BloemCJK15} take dependencies between components into account. In this setting, specifications are not always satisfiable by one component alone. Thus, a negotiation between the components is needed. While this yields more fine-grained decompositions, it produces a significant overhead that, as our experiments show, is often not necessary for common benchmarks.
Avoiding negotiation, dependency-based compositional synthesis~\cite{FinkbeinerP20} decomposes the system based on a dependency analysis of the specification. The analysis is more fine-grained than the one presented in this paper. Moreover, a weaker winning condition for synthesis, remorsefree dominance~\cite{DammF11}, is used. While this allows for smaller synthesis tasks since the specification can be decomposed further, both the dependency analysis and using a different winning condition produce a larger overhead than our approach.

The reactive synthesis tools Strix~\cite{MeyerStrix}, Unbeast~\cite{Ehlers11}, and Safety-First~\cite{SohailS13} decompose the given specification. Strix uses decomposition to find suitable automaton types for internal representation and to identify isomorphic parts of the specification. Unbeast and Safety-First in contrast, decompose the specification to identify safety parts.
All three tools do not perform independent synthesis tasks for the subspecifications. In fact, our experiments show that the scalability of Strix still improves notably with our algorithm.

Independent of~\cite{FinalVersion}, Mavridou et al. introduce a compositional realizability analysis of formulas given in FRET~\cite{GiannakopoulouP20a} that is based on similar ideas as our LTL decomposition algorithm~\cite{MavridouKGKPW21}. They only study the realizability of formulas but do not synthesize solutions. Optimized assumption handling cannot easily be integrated into their approach. 
For a detailed comparison of both approaches, we refer to~\cite{MavridouKGKPW21}.
The first version~\cite{FinalVersion} of our modular synthesis approach is already well-accepted in the synthesis community: Our LTL decomposition algorithm has been integrated into the new version~\cite{LTLsyntOptimized} of the synthesis tool ltlsynt~\cite{LTLsynt}.

\section{Preliminaries}

\paragraph{LTL.} 
Linear-time temporal logic~(LTL)~\cite{Pnueli77} is a specification language for linear-time properties. 
For a finite set $\Sigma$ of atomic propositions, the syntax of LTL is given by
$ \varphi, \psi ::= a ~ | ~ \true ~ | ~ \neg \varphi ~ | ~ \varphi \lor \psi ~ | ~ \varphi \land \psi ~ | ~ \Next \varphi ~ | ~ \varphi \Until \psi$, where $a \in \Sigma$.
We define the operators $\Eventually \varphi := \true \Until \varphi$ and $\Globally \varphi := \neg \Eventually \neg \varphi$ and use standard semantics.
The atomic propositions in $\varphi$ are denoted by $\propositions{\varphi}$, where every occurrence of $\true$ or $\false$ in $\varphi$ does not add any atomic propositions to $\propositions{\varphi}$.
The language~$\mathcal{L}(\varphi)$ of~$\varphi$ is the set of infinite words that satisfy $\varphi$.

\paragraph{Automata.}
For a finite alphabet $\Sigma$, a nondeterministic Büchi automaton~(NBA) is a tuple $\mathcal{A} = (Q,Q_0,\delta,F)$, where $Q$ is a finite set of states, $Q_0 \subseteq Q$ is a set of initial states, $\delta: Q \times \Sigma \times Q$ is a transition relation, and $F \subseteq Q$ is a set of accepting states.
Given an infinite word $\sigma = \sigma_1\sigma_2 \dots \in \Sigma^\omega$, a run of $\sigma$ on $\mathcal{A}$ is an infinite sequence $q_1 q_2 q_3 \dots \in Q^\omega$ of states where $q_1 \in Q_0$ and $(q_i,\sigma_i,q_{i+1}) \in \delta$ holds for all $i \geq 1$. 
A run is accepting if it contains infinitely many accepting states. $\mathcal{A}$ accepts a word $\sigma$ if there is an accepting run of~$\sigma$ on~$\mathcal{A}$.
The language $\mathcal{L}(\mathcal{A})$ of $\mathcal{A}$ is the set of all accepted words.
Two NBAs are equivalent if their languages are.
An LTL specification $\varphi$ can be translated into an equivalent NBA $\mathcal{A}_\varphi$ with a single exponential blow up~\cite{KupfermanV05}.

\paragraph{Implementations and Counterstrategies.}
An implementation of a system with inputs~$\inputs$, outputs~$\outputs$, and variables $V = \inputs \cup \outputs$ is a function $f : (2^V)^* \times 2^\inputs \rightarrow 2^\outputs$ mapping a history of variables and the current input to outputs.
An infinite word $\sigma = \sigma_1 \sigma_2 \dots \in (2^V)^\omega$ is compatible with an implementation $f$ if for all $n \in \mathbb{N}$, $f(\sigma_1 \dots \sigma_{n-1}, \sigma_n \cap \inputs) = \sigma_n \cap \outputs$ holds.
The set of all compatible words of $f$ is denoted by $\compatibleWords{f}$.
An implementation~$f$ realizes a specification~$s$ if $\sigma \in \mathcal{L}(s)$ holds for all $\sigma \in \compatibleWords{f}$. 
A specification  is called realizable if there exists an implementation realizing it.
If a specification is unrealizable, there is a counterstrategy $f^c:(2^V)^* \rightarrow 2^\inputs$ mapping a history of variables to inputs. 
An infinite word $\sigma = \sigma_1 \sigma_2 \dots \in (2^V)^\omega$ is compatible with $f^c$ if $f^c(\sigma_1 \dots \sigma_{n-1}) = \sigma_n \cap \inputs$ holds for all $n \in \mathbb{N}$. All compatible words of $f^c$ violate $s$, \ie, $\compatibleWords{f^c} \subseteq \overline{\mathcal{L}(s)}$.

\paragraph{Reactive Synthesis.}
Given a specification, reactive synthesis derives an implementation realizing it. 
For LTL specifications, synthesis is 2EXPTIME-complete~\cite{PnueliR89}.
In this paper, we use reactive synthesis as a black box procedure and thus we do not go into detail here. Instead, we refer the interested reader to~\cite{Finkbeiner16}.

\paragraph{Notation.}
Overloading notation, we use union and intersection on infinite words: For $\sigma = \sigma_1 \sigma_2 \dots \in (2^{\Sigma_1})^\omega$, $\sigma' = \sigma'_1 \sigma'_2 \dots \in (2^{\Sigma_2})^\omega$ with $\Sigma = \Sigma_1 \cup \Sigma_2$, we define $\sigma \cup \sigma' := (\sigma_1 \cup \sigma'_1) (\sigma_2 \cup \sigma'_2) \dots \in (2^{\Sigma})^\omega$.
For $\sigma$ as above and a set~$X$, let $\sigma \cap X := (\sigma_1 \cap X) (\sigma_2 \cap X) \dots \in (2^X)^\omega$.

\section{Modular Synthesis}

In this section, we introduce a modular synthesis algorithm that divides the synthesis task into independent subtasks by splitting the specification into several subspecifications.
The decomposition algorithm has to ensure that the synthesis tasks for the subspecifications can be solved independently and that their results are non-contradictory, \ie, that they can be combined into an implementation satisfying the initial specification.
Note that when splitting the specification, we assign a set of relevant in- and output variables to every subspecification. The corresponding synthesis subtask is then performed on these variables.

\begin{algorithm}[t]
	\SetKwInput{KwData}{Input}
	\SetKwInOut{KwResult}{Output}
	\SetKw{KwBy}{by}
	
	\KwData{\var{s}: Specification, \var{inp}: List Variable, \var{out}: List Variable}
	\KwResult{\var{realizable}: Bool, \var{implementation}: $\mathcal{T}$}
	\var{subspecifications} $\leftarrow$ decompose(\var{s}, \var{inp}, \var{out})\label{alg:compositional_synthesis:decompose} \\
	\var{sub\_results} $\leftarrow$ map synthesize \var{subspecifications}\label{alg:compositional_synthesis:synthesize} \\
	\ForEach{\upshape{(\var{real},\var{strat}) $\in$ \var{sub\_results}}}{	
		\If{\upshape{! \var{real}}}{
			\var{impl} $\leftarrow$ extendCounterStrategy(\var{strat}, \var{s})\label{alg:compositional_synthesis:counterstrategy} \\
			\Return{\upshape{($\bot$, \var{impl})}}
		}
	}
	\var{impls} $\leftarrow$ map second \var{sub\_results} \\
	\Return{\upshape{($\top$, compose \var{impls})}}
	\caption{Modular Synthesis}\label{alg:compositional_synthesis}
\end{algorithm}

\Cref{alg:compositional_synthesis} describes this modular synthesis approach. First, the specification is decomposed into a list of subspecifications
using an adequate decomposition algorithm.
Then, the synthesis tasks for all subspecifications are solved. If a subspecification is unrealizable, its counterstrategy is extended to a counterstrategy for the whole specification. This construction is given in \Cref{def:counterstrategy}.
Otherwise, the implementations of the subspecifications are composed.

Intuitively, the behavior of the counterstrategy of an unrealizable subspecification~$s_i$ violates the full specification~$s$ as well. A counterstrategy for the full specification, however, needs to be defined on all variables of~$s$, \ie, also on the variables that do not occur in~$s_i$.
Thus, we extend the counterstrategy for $\varphi_i$ such that it ignores outputs outside of $s_i$ and produces an arbitrary valuation of the input variables outside of $s_i$:

\begin{definition}[Counterstrategy Extension\label{def:counterstrategy}]
	Let $s$ be a specification with $\mathcal{L}(s) \subseteq (2^V)^\omega$.
    Let $V_1, V_2 \subset V$ with $V_1 \cup V_2 = V$ and $V_1 \cap V_2 \subseteq \inputs$. Let $s_1, s_2$ be subspecifications of $s$ with $\mathcal{L}(s_1) \subseteq (2^{V_1})^\omega$, $\mathcal{L}(s_2) \subseteq (2^{V_2})^\omega$ such that $\mathcal{L}(s_1) \pc \mathcal{L}(s_1)  = \mathcal{L}(s)$.
    Let $s_1$ be unrealizable and let $f^c_1: (2^{V_1})^* \rightarrow 2^{\inputs \cap V_1}$ be a counterstrategy for~$s_1$.
    We construct a counterstrategy $f^c: (2^V)^* \rightarrow 2^\inputs$ from~$f^c_1$ for $s$: $f^c(\sigma) = f^c_1(\sigma \cap V_1) \cup \mu$, where $\mu \in 2^{\inputs \setminus V_1}$ is an arbitrary valuation of the input variables outside of $V_1$.
\end{definition}

The counterstrategy for the full specification constructed as in \Cref{def:counterstrategy} then indeed fulfills the condition of a counterstrategy for the full specification, \ie, all of its compatible words violate the full specification:

\begin{lemma}\label{lem:extension_counterstrategy}
	Let $s$ be a specification with $\mathcal{L}(s) \subseteq (2^V)^\omega\!$.
    Let $V_1, V_2 \subset V$ with $V_1 \cup V_2 = V$, $V_1 \cap V_2 \subseteq \inputs$.
    Let $s_1, s_2$ be specifications with $\mathcal{L}(s_1) \subseteq (2^{V_1})^\omega$, $\mathcal{L}(s_2) \subseteq (2^{V_2})^\omega$ and $\mathcal{L}(s_1) \pc \mathcal{L}(s_1)  = \mathcal{L}(s)$.
    Let $f^c_1: (2^{V_1})^* \rightarrow 2^{\inputs \cap V_1}$ be a counterstrategy for $s_1$. The function $f^c$ constructed as in \Cref{def:counterstrategy} from $f^c_i$ is a counterstrategy for $s$.
\end{lemma}
\begin{proof}
    Let $\sigma \in \compatibleWords{f^c}$. Then $f^c(\sigma_1 \dots \sigma_{n-1}) = \sigma_n \cap \inputs$ for all $n \in \mathbb{N}$ and hence, by construction of $f^c$, we have $f^c_1(\restrict{\sigma_1 \dots \sigma_{n-1}}{V_1})= \sigma_n \cap (\inputs \cap V_1)$. Thus, $\sigma \cap V_1 \in \compatibleWords{f^c_1}$ follows. 
    Since $f^c_1$ is a counterstrategy for $s_1$, we have $\compatibleWords{f^c_1} \subseteq \overline{\mathcal{L}(s_1)}$. Hence, $\sigma \cap V_1 \in \overline{\mathcal{L}(s_1)}$.
    By assumption, $\mathcal{L}(s_1) \pc \mathcal{L}(s_2) = \mathcal{L}(s)$ and thus $(\restrict{\sigma}{V_1}) \cup \sigma' \not\in \mathcal{L}(s)$ for any infinite word $\sigma' \in (2^{V_2})^\omega$. 
    Thus, in particular, $(\restrict{\sigma}{V_1}) \cup (\restrict{\sigma}{V_2}) \not\in \mathcal{L}(s)$ holds. 
    Since $V_1 \cup V_2 = V$, $(\restrict{\sigma}{V_1}) \cup (\restrict{\sigma}{V_2}) = \sigma$ follows. Thus, $\sigma \notin \mathcal{L}(s)$. Hence, for all $\sigma \in \compatibleWords{f^c}$, $\sigma \not \in \mathcal{L}(s)$ and thus $\compatibleWords{f^c} \subseteq \overline{\mathcal{L}(s)}$. Thereforde, $f^c$ is a counterstrategy for~$s$.\qed
\end{proof}

Soundness and completeness of modular synthesis depend on three requirements: Equirealizability of the initial specification and the subspecifications, non-con\-tra\-dic\-to\-ry composability of the subresults, and satisfaction of the initial specification by the parallel composition of the subresults.
Intuitively, these requirements are met if the decomposition algorithm neither introduces nor drops parts of the system specification and if it does not produce subspecifications that allow for contradictory implementations.
To obtain composability of the subresults, the implementations need to agree on shared variables. We ensure this by assigning disjoint sets of output variables to the synthesis subtasks: Since every subresult only defines the behavior of the assigned output variables, the implementations are non-contradictory. 
Since the language alphabets of the subspecifications thus differ, the composition of their languages is non-contradictory:

\begin{definition}[Language Composition]
	Let $L_1$, $L_2$ be languages over $2^{\Sigma_1}$ and $2^{\Sigma_2}$, respectively. The \emph{non-contradictory composition of $L_1$ and~$L_2$} is given by $L_1 \! \pc L_2 \! = \! \{ \sigma_1 \cup \sigma_2 \mid \sigma_1 \! \in \! L_1 \land \sigma_2 \! \in \! L_2 \land \restrict{\sigma_1}{\Sigma_2} = \restrict{\sigma_2}{\Sigma_1} \}$.
\end{definition}

The satisfaction of the initial specification by the composed subresults can be guaranteed by requiring the subspecifications to be independent sublanguages:

\begin{definition}[Independent Sublanguages]
	Let $L \subseteq (2^\Sigma)^\omega$, $L_1 \subseteq (2^{\Sigma_1})^\omega$, and $L_2 \subseteq (2^{\Sigma_2})^\omega$ be languages with $\Sigma_1, \Sigma_2 \subseteq \Sigma$ and $\Sigma_1 \cup \Sigma_2 = \Sigma$. Then, $L_1$ and $L_2$ are \emph{independent sublanguages} of $L$ if $L_1 \pc L_2 = L$ holds.
\end{definition}

From these two requirements, \ie, the subspecifications form non-contradictory and independent sublanguages, equirealizability of the initial specification and the subspecifications follows:

\begin{theorem}\label{thm:equisynthesizeability_independent_sublanguages}
	Let $s$, $s_1$, and $s_2$ be specifications with $\mathcal{L}(s) \subseteq (2^V)^\omega$, $\mathcal{L}(s_1) \subseteq (2^{V_1})^\omega$, $\mathcal{L}(s_2) \subseteq (2^{V_2})^\omega$. 
	Recall that $I \subseteq V$ is the set of input variables.
	If $V_1 \cap V_2 \subseteq I$ and $V_1 \cup V_2 = V$ hold, and $\mathcal{L}(s_1)$ and $\mathcal{L}(s_2)$ are independent sublanguages of~$\mathcal{L}(s)$, then $s$ is realizable if, and only if, both $s_1$ and $s_2$ are realizable.
\end{theorem}
\begin{proof}
	First, suppose that $s_1$ and $s_2$ are realizable. Let $f_1: (2^{V_1})^* \times 2^{\inputs \cap V_1} \rightarrow 2^{\outputs \cap V_1}$, $f_2: (2^{V_2})^* \times 2^{\inputs \cap V_2} \rightarrow 2^{\outputs \cap V_2}$ be implementations realizing $s_1$ and $s_2$, respectively.
    We construct an implementation $f: (2^V)^* \times 2^\inputs \rightarrow 2^\outputs$ from $f_1$ and $f_2$: $f(\sigma,\inp{i}) := f_1(\restrict{\sigma}{V_1},\inp{i}\cap V_1) \cup f_2(\restrict{\sigma}{V_2}, \inp{i} \cap V_2)$. 	
    Let $\sigma \in \compatibleWords{f}$.
    Hence, $f((\sigma_1 \dots \sigma_{n-1}), \sigma_n \cap \inputs) = \sigma_n \cap \outputs$ for all $n \in \mathbb{N}$. 
    Let $\sigma' \in (2^{V_1})^\omega$, $\sigma'' \in (2^{V_2})^\omega$ be sequences with $\sigma'_n \cap \outputs =  f_1((\restrict{\sigma_1 \dots \sigma_{n-1}}{V_1}), \sigma_n \cap (\inputs \cap V_1))$ and $\sigma''_n \cap \outputs =  f_2((\restrict{\sigma_1 \dots \sigma_{n-1}}{V_2}), \sigma_n \cap (\inputs \cap V_2))$, respectively, for all $n \in \mathbb{N}$.
    Then, $\sigma'_n \cup \sigma''_n = \sigma_n \cap \outputs$ for all $n \in \mathbb{N}$ follows by construction of $f$ and thus $\sigma = \sigma' \cup \sigma''$ holds.
    Further, $\sigma' \in \compatibleWords{f_1}$ and $\sigma'' \in \compatibleWords{f_2}$ and thus, since~$s_1$ and $s_2$ are realizable by assumption, $\sigma' \in \mathcal{L}(s_1)$ and $\sigma'' \in \mathcal{L}(s_2)$. Since $\mathcal{L}(s_1)$ and $\mathcal{L}(s_2)$ are independent sublanguages by assumption, $\mathcal{L}(s_1) \pc \mathcal{L}(s_2) = \mathcal{L}(s)$ holds. Hence, by definition of language composition, $\sigma_1 \cup \sigma_2 \in \mathcal{L}(s)$ follows and thus, $\sigma \in \mathcal{L}(s)$ holds.
    Hence, for all $\sigma \in \compatibleWords{f}$, $\sigma \in \mathcal{L}(s)$ and therefore $f$ realizes $s$.
    
	Second, let $s_i$ is unrealizable for some $i \in \{1,2\}$ and let $f^c_i: (2^V)^* \rightarrow 2^{\inputs \cap V_1}$ be a counterstrategy for~$s_i$.
    We construct a counterstrategy $f^c: (2^V)^* \rightarrow 2^\inputs$ from $f^c_i$ as described in \Cref{def:counterstrategy}. By \Cref{lem:extension_counterstrategy},~$f^c$ is a counterstrategy for $s$. Thus, $s$ is unrealizable.\qed
\end{proof}

The soundness and completeness of \Cref{alg:compositional_synthesis} for adequate decomposition algorithms now follows directly with \Cref{thm:equisynthesizeability_independent_sublanguages} and the properties of such algorithms described above: They produce subspecifications that (1) do not share output variables and that (2) form independent sublanguages of the initial specification. 

\begin{theorem}
\label{thm:soundness_completeness}
	Let $s$ be a specification. Moreover, let $\mathcal{S} = \{s_1, \dots, s_k\}$ be a set of subspecifications of $s$ with $\mathcal{L}(s_i) \subseteq (2^{V_i})^\omega$ such that $\bigcup_{1 \leq i \leq k} V_i = V$, $V_i \cap V_j \subseteq \inputs$ for $1 \leq i,j \leq k$ with $i \neq j$, and such that $\mathcal{L}(s_1), \dots, \mathcal{L}(s_k)$ are independent sublanguages of $\mathcal{L}(s)$.
	If $s$ is realizable, \Cref{alg:compositional_synthesis} yields an implementation realizing $s$.
	Otherwise, \Cref{alg:compositional_synthesis} yields a counterstrategy for~$s$.
\end{theorem}
\begin{proof}
	First, let $s$ be realizable. Then, by applying \Cref{thm:equisynthesizeability_independent_sublanguages} recursively, it follows that $s_i$ is realizable for all $s_i \in \mathcal{S}$. Since $V_i \cap V_j \subseteq \inputs$ holds for any $s_i,s_j \in \mathcal{S}$ with $i \neq j$, the implementations realizing $s_1, \dots, s_k$ are non-contradictory. Hence, \Cref{alg:compositional_synthesis} returns their composition: Implementation~$f$. 
	Since $V_1 \cup \dots \cup V_k = V$,~$f$ defines the behavior of all outputs. 
	By construction, $f$ realizes all $s_i \in \mathcal{S}$. Since the $\mathcal{L}(s_i)$ are non-contradictory, independent sublanguages of $\mathcal{L}(s)$, $f$ thus realizes~$s$.
	
	Next, let $s$ be unrealizable. Then, by applying \Cref{thm:equisynthesizeability_independent_sublanguages} recursively, $s_i$ is unrealizable for some $s_i \in \mathcal{S}$. Thus, \Cref{alg:compositional_synthesis} returns the extension of $s_i$'s counterstrategy to a counterstrategy for the full specification. Its correctness follows with \Cref{lem:extension_counterstrategy}.\qed
\end{proof}

\section{Decomposition of Büchi Automata}\label{sec:automata}

To ensure soundness and completeness of modular synthesis, a specification decomposition algorithm needs to meet the language-based adequacy conditions of \Cref{thm:equisynthesizeability_independent_sublanguages}.
In this section, we lift these conditions from the language level to nondeterministic Büchi automata and present a decomposition algorithm for specifications given as NBAs on this basis.
Since the algorithm works directly on NBAs and not on their languages, we consider their composition instead of the composition of their languages:
Let $\mathcal{A}_1 = (Q_1,Q^1_0,\delta_1,F_1)$ and $\mathcal{A}_2 = (Q_2,Q^2_0,\delta_2,F_2)$ be NBAs over $2^{V_1}$, $2^{V_2}$, respectively. The \emph{parallel composition of $\mathcal{A}_1$ and $\mathcal{A}_2$} is defined by the NBA $\mathcal{A}_1 \pc \mathcal{A}_2 = (Q,Q_0,\delta,F)$ over~$2^{V_1 \cup V_2}$ with $Q = Q_1 \times Q_2$, $Q_0 = Q^1_0 \times Q^2_0$, $((q_1,q_2), \inp{i}, (q'_1,q'_2)) \in \delta$ if, and only if, $(q_1,\inp{i} \cap V_1,q'_1) \in \delta_1$ and $(q_2,\inp{i} \cap V_2,q'_2) \in \delta_2$, and $F = F_1 \times F_2$.
The parallel composition of NBAs reflects the composition of their languages:

\begin{lemma}\label{lem:correctness_parallel_composition_automata}
	Let $\mathcal{A}_1$ and $\mathcal{A}_2$ be NBAs over alphabets $2^{V_1}\!$ and $2^{V_2}\!$. Then, $\mathcal{L}(\mathcal{A}_1 \pc \mathcal{A}_2) = \mathcal{L}(\mathcal{A}_1) \pc \mathcal{L}(\mathcal{A}_2)$ holds.
\end{lemma}
\begin{proof}
	First, let $\sigma \in \mathcal{L}(\mathcal{A}_1 \pc \mathcal{A}_2)$. Then, $\sigma$ is accepted by $\mathcal{A}_1 \pc \mathcal{A}_2$. Hence, by definition of automaton composition, for $i \in \{1,2\}$, $\restrict{\sigma}{V_i}$ is accepted by~$\mathcal{A}_i$. Thus, $\restrict{\sigma}{V_i} \in \mathcal{L}(\mathcal{A}_i)$.
	Since $\restrict{(\restrict{\sigma}{V_1})}{V_2} = \restrict{(\restrict{\sigma}{V_2})}{V_1}$, we have $(\restrict{\sigma}{V_1}) \cup (\restrict{\sigma}{V_2}) \in \mathcal{L}(\mathcal{A}_1) \pc \mathcal{L}(\mathcal{A}_2)$.
	By definition of automaton composition, $\sigma \in (2^{V_1 \cup V_2})^\omega$ and thus $\sigma = (\restrict{\sigma}{V_1}) \cup (\restrict{\sigma}{V_2})$. Hence, $\sigma \in \mathcal{L}(\mathcal{A}_1) \pc \mathcal{L}(\mathcal{A}_2)$.
	
	Next, let $\sigma \in \mathcal{L}(\mathcal{A}_1) \pc \mathcal{L}(\mathcal{A}_2)$. Then, for $\sigma_1 \in (2^{V_1})^\omega$, $\sigma_2 \in (2^{V_2})^\omega$ with $\sigma = \sigma_1 \cup \sigma_2$, we have $\sigma_i \in \mathcal{L}(\mathcal{A}_i)$ for $i \in \{1,2\}$ and $\restrict{\sigma_1}{V_2} = \restrict{\sigma_2}{V_1}$.
	Hence,~$\sigma_i$ is accepted by $\mathcal{A}_i$. Thus, by definition of automaton composition and since $\sigma_1$ and $\sigma_2$ agree on shared variables, $\sigma_1 \cup \sigma_2$ is accepted by $\mathcal{A}_1 \pc \mathcal{A}_2$. Thus, $\sigma_1 \cup \sigma_2 \in \mathcal{L}(\mathcal{A}_1 \pc \mathcal{A}_2)$ and hence $\sigma \in \mathcal{L}(\mathcal{A}_1 \pc \mathcal{A}_2)$ holds.\qed
\end{proof}

Using the above lemma, we can formalize the independent sublanguage criterion on NBAs directly: Two automata $\mathcal{A}_1$, $\mathcal{A}_2$ are \emph{independent subautomata} of $\mathcal{A}$ if $\mathcal{A} = \mathcal{A}_1 \pc \mathcal{A}_2$.
To apply \Cref{thm:equisynthesizeability_independent_sublanguages}, the alphabets of the subautomata may not share output variables.
Our decomposition algorithm achieves this by constructing the subautomata from the initial automaton by projecting to disjoint sets of outputs.
Intuitively, the projection to a set $X$ abstracts from the variables outside of $X$. Hence, it only captures the parts of the initial specification concerning the variables in $X$. Formally:
Let $\mathcal{A} = (Q,Q_0,\delta,F)$ be an NBA over alphabet~$2^V$ and let $X \subset V$. 
The \emph{projection of $\mathcal{A}$ to $X$} is the NBA $\project{\mathcal{A}}{X} = (Q,Q_0,\pi_X(\delta),F)$ over~$2^X$ with $\pi_X(\delta) = \{ (q,a,q') \mid \exists~ b \in 2^{V \setminus X}.~(q,a \cup b,q')\in\delta\}$.

\begin{algorithm}[t]
	\SetKwInput{KwData}{Input}
	\SetKwInOut{KwResult}{Output}
	\SetKw{KwBy}{by}
	
	\KwData{$\mathcal{A}$: NBA, \var{inp}: List Variable, \var{out}: List Variable}
	\KwResult{\var{subautomata}: List (NBA, List Variable, List Variable)}
	
	\If{\upshape{isNull \var{checkedSubsets}}}{
		\var{checkedSubsets} $\leftarrow$ $\emptyset$
	}
	\var{subautomata} $\leftarrow$ [($\mathcal{A}$, \var{inp}, \var{out})] \\
	\ForEach{\upshape{\var{X} $\subset$ \var{out}}}{\label{alg:automaton-based_decomposition:guess}
		\var{Y} $\leftarrow$ \var{out}$\setminus$\var{X} \\
		\If{\upshape{\var{X} $\not\in$ \var{checkedSubsets} $\land$ \var{Y} $\not\in$ \var{checkedSubsets}}}{\label{alg:automaton-based_decomposition:unchecked}
			$\mathcal{A}_\var{X}$ $\leftarrow$ $\project{\mathcal{A}}{\var{X} \cup \var{inp}}$ \\
			$\mathcal{A}_\var{Y}$ $\leftarrow$ $\project{\mathcal{A}}{\var{Y} \cup \var{inp}}$ \\
			\If{\upshape{$\mathcal{L}(\mathcal{A}_\var{X}$ $\pc$ $\mathcal{A}_\var{Y})$~ $\subseteq$ $\mathcal{L}(\mathcal{A})$}}{\label{alg:automaton-based_decomposition:if}
				\var{subautomata} $\leftarrow$ decompose($\mathcal{A}_\var{X}$, \var{inp}, \var{X}) $++$ decompose($\mathcal{A}_\var{Y}$, \var{inp}, \var{Y})\label{alg:automaton-based_decomposition:add}\\
				break
			}
		}
		\var{checkedSubsets} $\leftarrow$ \var{checkedSubsets} $\cup$ $\{\var{X},\var{Y}\}$\label{alg:automaton-based_decomposition:store}
	}
	\Return{\upshape{\var{subautomata}}}
	\caption{Automaton Decomposition}\label{alg:automaton-based_decomposition}
\end{algorithm}

The decomposition algorithm for NBAs is described in \Cref{alg:automaton-based_decomposition}. It is a recursive algorithm that, starting with the initial automaton $\mathcal{A}$, guesses a subset~$\var{X}$ of the output variables $\var{out}$. It abstracts from the output variables outside of $\var{X}$ by building the projection $\mathcal{A}_\var{X}$ of~$\mathcal{A}$ to $\var{X} \cup \var{inp}$, where $\var{inp}$ is the set of input variables. Similarly, it builds the projection $\mathcal{A}_\var{Y}$ of~$\mathcal{A}$ to $\var{Y} := (\var{out} \setminus \var{X}) \cup \var{inp}$. By construction of $\mathcal{A}_\var{X}$ and $\mathcal{A}_\var{Y}$ and since both $\var{X} \cap \var{Y} = \emptyset$ and $\var{X} \cup \var{Y} = \var{out}$ hold, we have $\mathcal{L}(\mathcal{A}) \subseteq \mathcal{L}(\mathcal{A}_\var{X}$ $\pc$ $\mathcal{A}_\var{Y})$. Hence, if $\mathcal{L}(\mathcal{A}_\var{X}$ $\pc$ $\mathcal{A}_\var{Y}) \subseteq \mathcal{L}(\mathcal{A})$ holds, then $\mathcal{A}_\var{X}$ $\pc$ $\mathcal{A}_\var{Y}$ is equivalent to $\mathcal{A}$ and therefore $\mathcal{L}(\mathcal{A}_\var{X})$ and $\mathcal{L}(\mathcal{A}_\var{Y})$ are independent sublanguages of $\mathcal{L}(\mathcal{A})$. Thus, since $\var{X}$ and~$\var{Y}$ are disjoint and therefore $\mathcal{A}_\var{X}$ and $\mathcal{A}_\var{Y}$ do not share output variables, $\mathcal{A}_\var{X}$ and $\mathcal{A}_\var{Y}$ are a valid decomposition of~$\mathcal{A}$.
The subautomata are then decomposed recursively. If no further decomposition is possible, the algorithm returns the subautomata.
By only considering unexplored subsets of output variables, no subset combination $\var{X}, \var{Y}$ is checked twice.

\begin{figure}
	\centering
		\begin{tikzpicture}[>=latex,shorten >=0pt,auto,->,initial text = ,node distance=1cm,thin,every edge/.style={draw,font=\normalsize}]
		
			\node[state,initial]		(q0)		at (0,0)		{$q_0$};
			\node[state]				(q1)		at (-3,-3)	{$q_1$};
			\node[state,accepting]	(q2)		at (0,-2.3)	{$q_2$};
			\node[state]				(q3)		at (3,-3)	{$q_3$};

			\path 	(q0) edge [loop right]		node[align=left]		{\small$(\neg i \land \neg o_1)$ \\ \small$\lor~(\neg o_1 \land o_2)$} (q0)
						 edge					node[align=left,pos=0.6]	{\small$(\neg i \land o_1)$  \\ \small$\lor~ (o_1 \land o_2)$} (q2)
						 edge [bend left=20,sloped]		node[above,pos=0.65]		{\small$i \land o_1 \land \neg o_2$} (q3)
						 edge [bend right=20,sloped]				node[above]		{\small$i \land \neg o_1 \land \neg o_2$} (q1)
					(q1)	 edge [bend right=20,sloped]		node		{\small$\neg o_1 \land o_2$} (q0)
						 edge [loop below]				node		{\small$\neg o_1 \land \neg o_2$} (q1)
						 edge [sloped]					node	[below]	{\small$o_1 \land o_2$} (q2)
						 edge [bend right=28]			node	[below]	{\small$o_1 \land \neg o_2$} (q3)
					(q2) edge[loop below]				node		{\small$\neg i \lor o_2$} (q2)
						 edge [bend left=20,sloped]		node		{\small$i \land \neg o_2$} (q3)
					(q3) edge [bend left=20,sloped]		node	[below]	{\small$o_2$} (q2)
						 edge[loop below]				node		{\small$\neg o_2$} (q3);
		\end{tikzpicture}
		\caption{NBA $\mathcal{A}$ for $\varphi = \protect\Eventually o_1 \land \protect\Globally(i \rightarrow \protect\Eventually o_2)$. Accepting states are marked with double circles.}\label{fig:initial_automaton_example}
\end{figure}

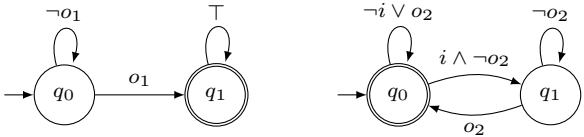
\begin{figure}[t]
	\centering
	\begin{subfigure}[t]{0.23\textwidth}
	\centering
	\begin{tikzpicture}[>=latex,shorten >=0pt,auto,->,initial text = ,node distance=1cm,thin,every edge/.style={draw,font=\normalsize}]
			\path[use as bounding box] (-0.8,1.25) rectangle (2.4,-0.55);
		
			\node[state,initial]		(q0)		at (0,0)		{$q_0$};
			\node[state,accepting]	(q1)		at (2,0)		{$q_1$};
			
			\path	(q0) edge [loop above]		node		{\small$\neg o_1$}	(q0)
						 edge					node		{\small$o_1$}		(q1)
					(q1) edge [loop above]		node		{\small$\top$}		(q1);	
	\end{tikzpicture}
	\caption{Minimization of $\project{\mathcal{A}}{V_1}$.}\label{fig:projection_1_min}
	\end{subfigure}
	\hfill
	\begin{subfigure}[t]{0.23\textwidth}
	\centering
	\begin{tikzpicture}[>=latex,shorten >=0pt,auto,->,initial text = ,node distance=1cm,thin,every edge/.style={draw,font=\normalsize}]
			\path[use as bounding box] (-0.8,1.25) rectangle (2.4,-0.55);
		
			\node[state,initial,accepting]		(q0)		at (0,0)		{$q_0$};
			\node[state]							(q1)		at (2,0)		{$q_1$};
			
			\path	(q0) edge [loop above]		node		{\small$\neg i \lor o_2$}	(q0)
						 edge [bend left=20]		node		{\small$i \land \neg o_2$}	(q1)
					(q1) edge [loop above]		node		{\small$\neg o_2$}			(q1)
						 edge [bend left=20]		node		{\small$o_2$}				(q0);
	\end{tikzpicture}
	\caption{Minimization of $\project{\mathcal{A}}{V_2}$.}\label{fig:projection_2_min}
	\end{subfigure}
	\caption{Minimized NBAs for the projections $\project{\mathcal{A}}{V_1}$ and $\project{\mathcal{A}}{V_2}$ of the NBA $\mathcal{A}$ from \Cref{fig:initial_automaton_example} to the sets of variables $V_1 = \{i,o_1\}$ and $V_2 = \{i,o_2\}$, respectively. Accepting states are marked with double circles.}\label{fig:projections_example_minimized}
\end{figure}

As an example for the specification decomposition algorithm based on NBAs, consider the specification $\varphi=\Eventually o_1 \land \Globally (i \rightarrow \Eventually o_2)$ for inputs $\inputs = \{i\}$ and outputs $\outputs = \{o_1,o_2\}$. 
The NBA~$\mathcal{A}$ that accepts $\mathcal{L}(\varphi)$ is depicted in \Cref{fig:initial_automaton_example}. 
The (minimized) subautomata obtained with \Cref{alg:automaton-based_decomposition} are shown in \Cref{fig:projection_1_min,fig:projection_2_min}. Clearly, $V_1 \cap V_2 \subseteq \inputs$ holds. Moreover, their parallel composition is exactly $\mathcal{A}$ depicted in \Cref{fig:initial_automaton_example} and therefore their parallel composition accepts exactly those words that satisfy $\varphi$.
For a slightly modified specification $\varphi' = \Eventually o_1 \lor \Globally (i \rightarrow \Eventually o_2)$, however, \Cref{alg:automaton-based_decomposition} does not decompose the NBA $\mathcal{A}'$ with $\mathcal{L}(\mathcal{A}') = \mathcal{L}(\varphi')$ depicted in \Cref{fig:initial_automaton_example_2}: The only possible decomposition is $\var{X} = \{o_1\}$, $\var{Y} = \{o_2\}$ (or vice-versa), yielding NBAs $\mathcal{A}'_\var{X}$ and $\mathcal{A}'_\var{Y}$ that accept every infinite word. Clearly, $\mathcal{L}(\mathcal{A}'_\var{X} \pc \mathcal{A}'_\var{Y}) \not\subseteq \mathcal{L}(\mathcal{A}')$ since $\mathcal{L}(\mathcal{A}'_\var{X} \pc \mathcal{A}'_\var{Y}) = (2^{\inputs\cup\outputs})^\omega$ and hence $\mathcal{A}'_\var{X}$ and $\mathcal{A}'_\var{Y}$ are no valid decomposition.

\begin{figure}
	\centering
		\begin{tikzpicture}[>=latex,shorten >=0pt,auto,->,initial text = ,node distance=1cm,thin,every edge/.style={draw,font=\normalsize}]
		
			\node[state,initial]		(q0)		at (0,0)		{$q_0$};
			\node[state]				(q1)		at (-3.5,-2)	{$q_1$};
			\node[state,accepting]	(q2)		at (-1.75,-3.5)	{$q_2$};
			\node[state,accepting]	(q3)		at (3.5,-2)	{$q_3$};
			\node[state]				(q4)		at (1.75,-3.5)	{$q_4$};
			
			\path	(q0) edge				node	[above,left]	{\small$\neg o_1~$} (q1)
						 edge[sloped]		node[align=center,above]		{\small$(\neg i \land \neg o_1)\lor(\neg o_1 \land o_2)$} (q3)
						 edge				node		{\small$o_1$} (q2)
						 edge[sloped]		node		{\small$i \land \neg o_1 \land \neg o_2$} (q4)
					(q1) edge[loop above]	node		{\small$\neg o_1$} (q1)
						 edge				node		{\small$o_1$} (q2)
					(q2) edge[loop right]	node		{\small$\top$} (q2)
					(q3) edge[loop above]	node		{\small$\neg i \lor o_2$} (q3)
						 edge[bend left=20] node	[sloped,below]	{\small$i \land \neg o_2$} (q4)
					(q4) edge[bend left=20] node	[sloped]	{\small$o_2$} (q3)
						 edge[loop left]	node		{\small$\neg o_2$} (q4);	
		\end{tikzpicture}
		\caption{NBA $\mathcal{A}'$ for $\varphi' = \protect\Eventually o_1 \lor \protect\Globally(i \rightarrow \protect\Eventually o_2)$. Accepting states are marked with double circles.}\label{fig:initial_automaton_example_2}
\end{figure}

\Cref{alg:automaton-based_decomposition} ensures soundness and completeness of modular synthesis: The subspecifications do not share output variables and they are equirealizable to the initial specification.
This follows from the construction of the subautomata, \Cref{lem:correctness_parallel_composition_automata}, and \Cref{thm:equisynthesizeability_independent_sublanguages}:

\begin{theorem}\label{thm:correctness_automaton_decomposition}
	Let $\mathcal{A}$ be an NBA over alphabet $2^V$. \Cref{alg:automaton-based_decomposition} terminates on $\mathcal{A}$ with a set $\mathcal{S} = \{\mathcal{A}_1, \dots, \mathcal{A}_k\}$ of NBAs with $\mathcal{L}(\mathcal{A}_i) \subseteq (2^{V_i})^\omega$, where $V_i \cap V_j \subseteq \inputs$ for $1 \leq i,j \leq k$ with $i \neq j$, $V = \bigcup_{1 \leq i \leq k} V_i$, and $\mathcal{A}$ is realizable if, and only if, $\mathcal{A}_i$ is realizable for all $\mathcal{A}_i \in \mathcal{S}$.
\end{theorem}

\begin{proof}
    Clearly, there are NBAs that cannot be decomposed further, \eg, automata whose alphabet contains only one output variable. Thus, since there are only finitely many subsets of $\outputs$, \Cref{alg:automaton-based_decomposition} terminates.
    
    We show that the algorithm returns subspecifications that only share input variables, define all output variables of the system, and that are independent sublanguages of the initial specification by structural induction on the initial automaton:
    
    For any automaton $\mathcal{A}'$ that is not further decomposable, \Cref{alg:automaton-based_decomposition} returns a list  $\mathcal{S}'$ solely containing $\mathcal{A}'$. Clearly, the parallel composition of all automata in $\mathcal{S}'$ is equivalent to $\mathcal{A}'$ and the alphabets of the languages of the subautomata do not share output variables.
    
    Next, let $\mathcal{A}'$ be an NBA such that there exists a set \var{X} $\subset \var{out}$ with $\mathcal{L}(\project{\mathcal{A}'}{\var{X}\cup\var{inp}} \pc \project{\mathcal{A}'}{\var{Y} \cup \var{inp}}) \subseteq \mathcal{L}(\mathcal{A}')$, where $\var{Y} = \var{out} \setminus \var{X}$. 
    By construction of $\project{\mathcal{A}'}{\var{X}\cup\var{inp}}$ and $\project{\mathcal{A}'}{\var{Y}\cup\var{inp}}$, we have $(\mathcal{A}' \cap (\var{Z}\cup\var{inp})) \subseteq \project{\mathcal{A}'}{\var{Z}\cup\var{inp}}$ for $\var{Z} \in \{ \var{X},\var{Y} \}$. Since both $\var{X} \cap \var{Y} = \emptyset$ and $\var{X} \cup \var{Y} = \var{out}$ hold by construction of $\var{X}$ and $\var{Y}$, $(\var{X}\cup\var{inp}) \cap (\var{Y}\cup\var{inp}) \subseteq \var{inp}$ as well as $(\var{X}\cup\var{inp}) \cup (\var{Y}\cup\var{inp}) = \var{inp} \cup \var{out}$ follows. 
    Therefore, $\mathcal{L}(\mathcal{A}) \subseteq \mathcal{L}(\project{\mathcal{A}'}{\var{X}\cup\var{inp}} \pc \project{\mathcal{A}'}{\var{Y} \cup \var{inp}})$ holds and thus, $\project{\mathcal{A}'}{\var{X}\cup\var{inp}} \pc \project{\mathcal{A}'}{\var{Y} \cup \var{inp}} \equiv \mathcal{A}'$ follows.
    By induction hypothesis, the calls to the algorithm with $\project{\mathcal{A}'}{\var{X}\cup\var{inp}}$ and $\project{\mathcal{A}'}{\var{Y} \cup \var{inp}}$ return lists $\mathcal{S}'_\var{X}$ and $\mathcal{S}'_{\var{Y}}$, respectively, where the parallel composition of all automata in $\mathcal{S}'_\var{Z}$ is equivalent to $\project{\mathcal{A}'}{\var{Z} \cup \var{inp}}$ for $\var{Z} \in \{\var{X}, \var{Y}\}$.
    Thus, the parallel composition of all automata in the concatenation of $\mathcal{S}'_\var{X}$ and~$\mathcal{S}'_{\var{Y}}$ is equivalent to $\project{\mathcal{A}'}{\var{X}\cup\var{inp}} \pc \project{\mathcal{A}'}{\var{Y} \cup \var{inp}}$ and thus, by construction of \var{X}, to~$\mathcal{A}'$. Hence, their languages are independent sublanguages of $\mathcal{A}'$.
    Furthermore, by induction hypothesis, the alphabets of the automata in $\mathcal{S}'_\var{Z}$ do not share output variables for $\var{Z} \in \{\var{X}, \var{Y}\}$ and, by construction, they are subsets of the alphabet of $\project{\mathcal{A}'}{\var{Z}}$. Hence, since clearly $(\var{X} \cup \var{inp}) \cap ((\var{out} \setminus \var{X}) \cup \var{inp}) \subseteq \var{inp}$ holds, the alphabets of the automata in the concatenation of $\mathcal{S}'_\var{X}$ and $\mathcal{S}'_{\var{Y}}$ do not share output variables.
    Moreover, the union of the alphabets of the automata in $\mathcal{S}'_\var{Z}$ equals the alphabet of $\project{\mathcal{A}}{\var{Z} \cup \var{inp}}$ for $\var{Z} \in \{\var{X}, \var{Y}\}$ by induction hypothesis. Since clearly $\var{X} \cup \var{Y} = \var{out}$, it follows that the union of the alphabets of the automata in the concatenation of $\mathcal{S}'_\var{X}$ and $\mathcal{S}'_{\var{Y}}$ equals $\var{inp} \cup \var{out}$.
    
    Thus, $\bigcup_{1 \leq i \leq k} V_i = V$ and $V_i \cap V_j \subseteq \inputs$ for $1 \leq i,j \leq k$ with $i \neq j$. Moreover, $\mathcal{L}(\mathcal{A}_1), \dots, \mathcal{L}(\mathcal{A}_k)$ are independent sublanguages of $\mathcal{L}(\mathcal{A})$. Thus, by \Cref{thm:equisynthesizeability_independent_sublanguages}, $\mathcal{A}$ is realizable if, and only if, all $\mathcal{A}_i \in \mathcal{S}$ are realizable.\qed
\end{proof}

Since \Cref{alg:automaton-based_decomposition} is called recursively on every subautomaton obtained by projection, it directly follows that the nondeterministic Büchi automata contained in the returned list are not further decomposable:

\begin{theorem}
	Let $\mathcal{A}$ be an NBA and let $\mathcal{S}$ be the set of NBAs that $\Cref{alg:automaton-based_decomposition}$ returns on input $\mathcal{A}$. Then, for each $\mathcal{A}_i \in \mathcal{S}$ over alphabet $2^{V_i}$, there are no NBAs $\mathcal{A}'$, $\mathcal{A''}$ over alphabets $2^{V'}$ and $2^{V''}$ with $V_i = V' \cup V''$ such that $\mathcal{A}_i = \mathcal{A}' \pc \mathcal{A}''$ holds.
\end{theorem}

Hence, \Cref{alg:automaton-based_decomposition} yields \emph{perfect} decompositions and is semantically precise.
Yet, it performs several expensive automaton operations such as projection, composition, and language containment checks.
For large automata, this is infeasible.
For specifications given as LTL formulas, we thus present an approximate decomposition algorithm in the next section that does not yield non-decomposable subspecifications, but that is free of the expensive automaton operations.

\section{Decomposition of LTL Formulas}

An LTL specification can be decomposed by translating it into an equivalent NBA and by then applying \Cref{alg:automaton-based_decomposition}.
To circumvent expensive automaton operations, though, we introduce an approximate decomposition algorithm that, in contrast to \Cref{alg:automaton-based_decomposition}, does not necessarily find all possible decompositions.
In the following, we assume that $V = \propositions{\varphi}$ holds for the initial specification~$\varphi$. Note that any implementation for the variables in $\propositions{\varphi}$ can easily be extended to one for the variables in $V$ if $\propositions{\varphi} \subset V$ holds by ignoring the inputs in $\inputs \setminus \propositions{\varphi}$ and by choosing arbitrary valuations for the outputs in $\outputs \setminus \propositions{\varphi}$.

The main idea of the decomposition algorithm is to rewrite the initial LTL formula $\varphi$ into a conjunctive form $\varphi=\varphi_1 \land \dots \land \varphi_k$ with as many top-level conjuncts as possible by applying distributivity and pushing temporal operators inwards whenever possible. Then, we build subspecifications $\varphi_i$ consisting of subsets of the conjuncts. Each conjunct occurs in exactly one subspecification.
We say that conjuncts are \emph{independent} if they do not share output variables.
Given an LTL formula with two independent conjuncts, the languages of the conjuncts are independent sublanguages of the language of the whole formula:
\begin{lemma}\label{lem:independent_sublanguages_conjuncts}
	Let $\varphi = \varphi_1 \land \varphi_2$ be an LTL formula over atomic propositions $V$ with conjuncts $\varphi_1$ and $\varphi_2$ over $V_1$ and $V_2$, respectively, with $V_1 \cup V_2 \subseteq V$.
	Then, $\mathcal{L}(\varphi_1)$ and~$\mathcal{L}(\varphi_2)$ are independent sublanguages of $\mathcal{L}(\varphi)$.
\end{lemma}
\begin{proof}
	First, let $\sigma \in \mathcal{L}(\varphi)$. Then, $\sigma \in \mathcal{L}(\varphi_i)$ holds for all $i \in \{1,2\}$. Since $\propositions{\varphi_i} \subseteq V_i$ holds and since the satisfaction of $\varphi_i$ only depends on the valuations of the variables in $\propositions{\varphi_i}$, we have $\restrict{\sigma}{V_i} \in \mathcal{L}(\varphi_i)$. Since clearly $\restrict{(\restrict{\sigma}{V_1})}{V_2} = \restrict{(\restrict{\sigma}{V_2})}{V_1}$ holds, we have $(\restrict{\sigma}{V_1}) \cup (\restrict{\sigma}{V_2}) \in \mathcal{L}(\varphi_1) \pc \mathcal{L}(\varphi_2)$. Since $V_1 \cup V_2 = V$ holds by assumption, we have $\sigma = (\restrict{\sigma}{V_1}) \cup (\restrict{\sigma}{V_2})$ and hence $\sigma \in \mathcal{L}(\varphi_1) \pc \mathcal{L}(\varphi_2)$ follows.
	
	Next, let $\sigma \in \mathcal{L}(\varphi_1) \pc \mathcal{L}(\varphi_2)$. Then, there are words $\sigma_1 \in \mathcal{L}(\varphi_1)$, $\sigma_2 \in \mathcal{L}(\varphi_2)$ with $\restrict{\sigma_1}{V_2} = \restrict{\sigma_2}{V_1}$ and $\sigma = \sigma_1 \cup \sigma_2$. Since $\sigma_1$ and $\sigma_2$ agree on shared variables, $\sigma \in \mathcal{L}(\varphi_1)$ and $\sigma \in \mathcal{L}(\varphi_2)$. Hence, $\sigma \in \mathcal{L}(\varphi_1 \land \varphi_2)$.\qed
\end{proof}

Our decomposition algorithm then ensures that different subspecifications share only input variables by merging conjuncts that share output variables into the same subspecification. Then, equirealizability of the initial formula and the subformulas follows directly from \Cref{thm:equisynthesizeability_independent_sublanguages} and \Cref{lem:independent_sublanguages_conjuncts}:

\begin{corollary}\label{cor:equisynthesizeability_independent_conjuncts}
	Let $\varphi = \varphi_1 \land \varphi_2$ be an LTL formula over $V$ with conjuncts $\varphi_1$, $\varphi_2$ over $V_1$, $V_2$, respectively, with $V_1 \cup V_2 = V$ and $V_1 \cap V_2 \subseteq \inputs$. Then, $\varphi$ is realizable if, and only if, both $\varphi_1$ and $\varphi_2$ are realizable.
\end{corollary}

To determine which conjuncts of an LTL formula $\varphi = \varphi_1 \land \dots \land \varphi_n$ share variables, we build the \emph{dependency graph} $\depGraph{\varphi} = (V,E)$ of the output variables, where $V = \outputs$ and $(a,b) \in E$ if, and only if, $a \in \propositions{\varphi_i}$ and $b \in \propositions{\varphi_i}$ for some $1 \leq i \leq n$. Intuitively, outputs $a$ and $b$ that are contained in the same connected component of $\depGraph{\varphi}$ depend on each other in the sense that they either occur in the same conjunct or that they occur in conjuncts that are connected by other output variables. 
Hence, to ensure that subspecifications do not share output variables, conjuncts containing $a$ or $b$ need to be assigned to the same subspecification.
Output variables that are contained in different connected components, however, are not linked and therefore implementations for their requirements can be synthesized independently, \ie, with independent subspecifications.

\begin{algorithm}[t]
	\SetKwInput{KwData}{Input}
	\SetKwInOut{KwResult}{Output}
	\SetKw{KwBy}{by}
	
	\KwData{$\varphi$: LTL, \var{inp}: List Variable, \var{out}: List Variable}
	\KwResult{\var{specs}: List (LTL, List Variable, List Variable)}
	
	$\varphi$ $\leftarrow$ rewrite$(\varphi)$ \\
	\var{formulas} $\leftarrow$ removeTopLevelConjunction$(\varphi)$ \\
	\var{graph} $\leftarrow$ buildDependencyGraph($\varphi$, \var{out}) \\
	\var{components} $\leftarrow$ \var{graph}.connectedComponents() \\
    \var{specs} $\leftarrow$ new LTL[$|$\var{components}$|$+1] ~\tcp{initialized with true}
	\ForEach{\upshape{$\psi$ $\in$ \var{formulas}}}{		
		\var{propositions} $\leftarrow$ getProps$(\psi)$ \\
		\ForEach{\upshape{(\var{spec},\var{set}) $\in$ zip \var{specs} (\var{components} $++$ [\var{inp}])}}{\label{alg:rewriting-based_decomposition:zip}
			\If{\upshape{\var{propositions} $\cap$ \var{set} $\neq$ $\emptyset$}} {
				\var{spec}.And$(\psi)$\label{alg:rewriting-based_decomposition:add}\\
				break\label{alg:rewriting-based_decomposition:break}
			}
		}
	}
    \Return{\emph{map ($\lambda \varphi   \rightarrow$ ($\varphi$, inputs($\varphi$), outputs($\varphi$)))} \upshape{\var{specs}}}
	\caption{LTL Decomposition}\label{alg:rewriting-based_decomposition}
\end{algorithm}

\Cref{alg:rewriting-based_decomposition} describes how an LTL formula is decomposed into subspecifications. First, the formula is rewritten into conjunctive form. Then, the dependency graph is built and the connected components are computed. For each connected component as well as for all input variables, a subspecification is built by adding the conjuncts containing variables of the respective connected component or an input variable, respectively.
To also consider the input variables is necessary to assign every conjunct, including input-only ones, to at least one subspecification.
By construction, no conjunct is added to the subspecifications of two different connected components.
Yet, a conjunct could be added to both a subspecification of a connected component and the subspecification for the input-only conjuncts. This is circumvented by the \emph{break} in \Cref{alg:rewriting-based_decomposition:break}. Hence, every conjunct is added to exactly one subspecification.
To define the input and output variables for the synthesis subtasks, the algorithm assigns the inputs and outputs occurring in~$\varphi_i$ to the subspecification $\varphi_i$.
While restricting the inputs is not necessary for correctness, it may improve the runtime of the synthesis task.

As an example for the decomposition of LTL formulas, consider the specification $\varphi = \Eventually o_1 \land \Globally(i \rightarrow o_2)$ with $\inputs = \{i\}$ and $\outputs = \{o_1,o_2\}$ again. Since $\varphi$ is already in conjunctive form, no rewriting has to be performed. The two conjuncts of $\varphi$ do not share any variables and therefore the dependency graph $\mathcal{D}_\varphi$ does not contain any edges. Therefore, we obtain two subspecifications $\varphi_1 = \Eventually o_1$ and $\varphi_2 = \Globally(i \rightarrow o_2)$.


Soundness and completeness of modular synthesis with \Cref{alg:rewriting-based_decomposition} as a decomposition algorithm for LTL formulas follows directly from \Cref{cor:equisynthesizeability_independent_conjuncts} if the subspecifications do not share any output variables:

\begin{theorem}
	Let $\varphi$ be an LTL formula over $V$. Then, $\Cref{alg:rewriting-based_decomposition}$ terminates with a set $\mathcal{S}=\{\varphi_1, \dots, \varphi_k\}$ of LTL formulas on $\varphi$ with $\mathcal{L}(\varphi_i) \in (2^{V_i})^\omega$ such that $V_i \cap V_j \subseteq \inputs$ for $1 \leq i,j \leq k$ with $i \neq j$, $\bigcup_{1 \leq i \leq k} V_i = V$, and such that $\varphi$ is realizable, if, and only if, for all subspecifications $\varphi_i \in \mathcal{S}$, $\varphi_i$ is realizable.
\end{theorem}
\begin{proof}
	Since an output variable is part of exactly one connected component and since all conjuncts containing an output are contained in the same subspecification, every output is part of exactly one subspecification. Therefore, $V_i \cap V_j \subseteq \inputs$ holds for $1 \leq i,j \leq k$ with $i \neq j$.
	Moreover, the last component added in \Cref{alg:rewriting-based_decomposition:zip} contains all inputs. Hence, all variables that occur in a conjunct of $\varphi$ are featured in at least one subspecification. Thus, $\bigcup_{1\leq i \leq k} V_i = \propositions{\varphi}$ holds and hence, since $V = \propositions{\varphi}$ by assumption, $\bigcup_{1\leq i \leq k} V_i = V$ follows.
	Therefore, equirealizability of $\varphi$ and the formulas in $\mathcal{S}$ directly follows with \Cref{cor:equisynthesizeability_independent_conjuncts}.\qed
\end{proof}

While \Cref{alg:rewriting-based_decomposition} is simple and ensures soundness and completeness of modular synthesis, it strongly depends on the structure of the formula:
When rewriting formulas in assume-guarantee format, \ie, formulas of the form $\varphi = \bigwedge^m_{i=1} \varphi_i \rightarrow \bigwedge^n_{j=1} \psi_j$, to a conjunctive form, the conjuncts contain both assumptions $\varphi_i$ and guarantees~$\psi_j$. Hence, if $a,b \in \outputs$ occur in assumption~$\varphi_i$ and guarantee $\psi_j$, respectively, they are dependent. Thus, all conjuncts featuring $a$ or $b$ are contained in the same subspecification according to \Cref{alg:rewriting-based_decomposition}. Yet, $\psi_j$ might be realizable even without~$\varphi_i$. An algorithm accounting for this might yield further decompositions and thus smaller synthesis subtasks.

In the following, we present a criterion for dropping assumptions while maintaining equirealizability. Intuitively, we can drop an assumption $\varphi$ for a guarantee~$\psi$ if they do not share any variable. However, if $\varphi$ can be violated by the system, \ie, if $\neg \varphi$ is realizable, equirealizability is not guaranteed when dropping $\varphi$. For instance, consider the formula $\varphi = \Eventually(i_1 \land o_1) \rightarrow \Globally (i_2 \land o_2)$, where $\inputs = \{i_1,i_2\}$ and $\outputs = \{o_1,o_2\}$. Although assumption and guarantee do not share any variables, the assumption cannot be dropped: An implementation that never sets $o_1$ to $\true$ satisfies $\varphi$ but $\Globally(i_2 \land o_2)$ is not realizable.
Furthermore, dependencies between input variables may yield unrealizability if an assumption is dropped as information about the remaining inputs might get lost. For instance, in the formula $\varphi \rightarrow \psi$ with $\varphi \! = \! (\Globally i_1 \! \rightarrow i_2) \land (\neg\Globally i_1 \! \rightarrow i_3) \land (i_2 \! \leftrightarrow i_4) \land (i_3 \! \leftrightarrow \neg i_4)$ and $\psi = \Globally i_1 \leftrightarrow o$, where $\inputs = \{i_1,i_2,i_3,i_4\}$ and $\outputs = \{o\}$, no assumption can be dropped: Otherwise the information about the global behavior of $i_1$, which is crucial for the existence of an implementation, is incomplete.
These observations lead to the following criterion for safely dropping assumptions. 

\begin{lemma}\label{lem:assumption_dropping}
	Let $\varphi = (\varphi_1 \land \varphi_2) \rightarrow \psi$ be an LTL~formula with $\propositions{\varphi_1} \cap \propositions{\varphi_2} = \emptyset$, $\propositions{\varphi_2} \cap \propositions{\psi} = \emptyset$. Let~$\neg\varphi_2$ be unrealizable.
	Then, $\varphi_1 \rightarrow \psi$ is realizable if, and only if, $\varphi$ is realizable.
\end{lemma}
%

\begin{proof}
    Let $V_1 := \propositions{\varphi_1} \cup \propositions{\psi}$, $\inputs_1 := \inputs \cap V_1$, and $\outputs_1 := \outputs \cap V_1$.
    First, let $\varphi_1 \rightarrow \psi$ be realizable. Then there is an implementation $f_1: (2^{V_1})^* \times 2^{\inputs_1} \rightarrow 2^{\outputs_1}$ that realizes $\varphi_1 \rightarrow \psi$.
    From $f_1$, we construct a strategy $f:(2^V)^* \times 2^\inputs \rightarrow 2^\outputs$ as follows: Let $\mu \in 2^{O \setminus O_1}$ is an arbitrary valuation of the outputs outside of $O_1$.
 	Then, let $(\sigma, \inp{i}) := f_1(\restrict{\sigma}{V_1},\inp{i}\cap I_1) \cup \mu$.
    Let $\sigma \in \compatibleWords{f}$. Then we have $f(\sigma_1 \dots \sigma_{n-1}, \sigma_n \cap \inputs) = \sigma_n \cap \inputs$ for all $n \in \mathbb{N}$ and thus $f_1(\restrict{(\sigma_1 \dots \sigma_{n-1})}{V_1}, \sigma \cap I_1)= \sigma_n \cap (\inputs \cap V_1)$ follows by construction of $f$. Hence, $\restrict{\sigma}{V_1} \in \compatibleWords{f_1}$ holds and thus, since $f_1$ realizes $\varphi_1 \rightarrow \psi$ by assumption, $\restrict{\sigma}{V_1} \in \mathcal{L}(\varphi_1 \rightarrow \psi)$.
    Since $\propositions{\varphi_1} \cap \propositions{\varphi_2} = \emptyset$ and $\propositions{\varphi_2} \cap \propositions{\psi} = \emptyset$, we have $\propositions{\varphi_2} \cap V_1 = \emptyset$. Hence, the valuations of the variables in $\propositions{\varphi_2}$ do not affect the satisfaction of $\varphi_1 \rightarrow \psi$. Thus, we have $(\restrict{\sigma}{V_1}) \cup \sigma' \in \mathcal{L}(\varphi_1 \rightarrow \psi)$ for any $\sigma' \in (2^\propositions{\varphi_2})^\omega$. In particular, $(\restrict{\sigma}{V_1}) \cup (\restrict{\sigma}{\propositions{\varphi_2}}) \in \mathcal{L}(\varphi_1 \rightarrow \psi)$.
    Since $\propositions{\varphi} = V$ by assumption, $V = V_1 \cup \propositions{\varphi_2}$ holds and thus $(\restrict{\sigma}{V_1}) \cup (\restrict{\sigma}{\propositions{\varphi_2}}) = \sigma$.
    Hence, $\sigma \in \mathcal{L}(\varphi_1 \rightarrow \psi)$ holds and thus, since $\varphi_1 \rightarrow \psi$ implies $(\varphi_1 \land \varphi_2) \rightarrow \psi$, $\sigma \in \mathcal{L}(\varphi)$ follows. Hence, $f$ realizes $\varphi$.
    
    Next, let $(\varphi_1 \land \varphi_2) \rightarrow \psi$ be realizable. Then, there is an implementation $f: (2^V)^* \times 2^\inputs \rightarrow 2^\outputs$ that realizes $(\varphi_1 \land \varphi_2) \rightarrow \psi$.
    Since $\neg\varphi_2$ is unrealizable, there is a counterstrategy $f^c_2: (2^{\propositions{\varphi_2}})^* \rightarrow 2^{\inputs \cap \propositions{\varphi_2}}$ for $\neg \varphi_2$ and all words compatible with $f^c_2$ satisfy $\varphi_2$.
    Given a finite sequence $\eta \in (2^{V_1})^*$, let $\hat{\eta} \in (2^V)^*$ be the sequence obtained by lifting $\eta$ to $V$ using the output of~$f^c_2$. Formally, let $\hat{\eta} = h(\varepsilon,\eta)$, where $h: (2^V)^* \times (2^{V_1})^* \rightarrow (2^V)^*$ is a function defined by $h(\tau,\varepsilon) = \tau$ for the empty word~$\varepsilon$ and, when $\boldsymbol{\cdot}: V \times V^* \rightarrow V^*$ denotes concatenation, $h(\tau,s \boldsymbol{\cdot} \eta) = h(\tau \boldsymbol{\cdot} ((s \cap \inputs) \cup c \cup f(\tau, ((s \cap \inputs) \cup c) \cap \inputs)), \eta)$ with $c = f^c_2(\restrict{\tau}{\propositions{\varphi_2}})$.
    We construct an implementation $g: (2^{V_1})^* \times 2^{\inputs_1} \rightarrow 2^{\outputs_1}$ based on $f$ and $\hat{\eta}$ as follows: $g(\eta,\inp{i}) := f(\hat{\eta}, \inp{i} \cup (f^c_2(\hat{\eta}) \cap \inputs)) \cap \outputs_1$.
    Let $\sigma \in \compatibleWords{g}$. Let $\sigma_\mathit{f}$ be the corresponding infinite sequence obtained from $g$ when not restricting the output of $f$ to $\outputs_1$. Hence, $\sigma_\mathit{f} \cap V_1 = \sigma$.
    Clearly, by construction of $g$, we have $\sigma_\mathit{f} \in \compatibleWords{f}$ and hence, since~$f$ realizes~$\varphi$ by assumption, $\sigma_\mathit{f} \in \mathcal{L}(\varphi)$.
    Furthermore, we have $\sigma_\mathit{f} \in \mathcal{L}(\varphi_2)$ by construction of $g$ since $\hat{\eta}$ forces~$f$ to satisfy~$\varphi_2$.
    Hence, $\sigma_\mathit{f} \in \mathcal{L}(\varphi_1 \rightarrow \psi)$.
    Since $\varphi_2$ neither shares variables with $\varphi_1$ nor with $\psi$ by assumption, the satisfaction of $\varphi_1 \rightarrow \psi$ is not influenced by the variables outside of $V_1$. Thus, since we have $\sigma_\mathit{f} \cap V_1 = \sigma$ by construction, $\sigma \in \mathcal{L}(\varphi_1 \rightarrow \psi)$ follows.
    Hence, $g$ realizes $\varphi_1 \rightarrow \psi$.\qed
\end{proof}

By dropping assumptions, we are able to decompose LTL formulas of the form $\varphi = \bigwedge^m_{i=1} \varphi_i \rightarrow \bigwedge^n_{j=1} \psi_j$ in further cases: 
We rewrite $\varphi$ to $\bigwedge^n_{j=1}(\bigwedge^m_{i=1} \varphi_i \rightarrow \psi_j)$ and then drop assumptions for the individual guarantees. If the resulting subspecifications only share input variables, they are equirealizable to $\varphi$. 

\begin{theorem}\label{thm:ltl_decomposition_with_assumptions}
	Let $\varphi = (\varphi_1 \land \varphi_2 \land \varphi_3) \rightarrow (\psi_1 \land \psi_2)$ be an LTL formula over $V$, where $\propositions{\varphi_3} \subseteq \inputs$ and $\propositions{\psi_1} \cap \propositions{\psi_2} \subseteq \inputs$. 
	Let $\propositions{\varphi_i} \cap \propositions{\varphi_j} = \emptyset$ for $i,j \in \{1,2,3\}$ with $i \neq j$, and $\propositions{\varphi_i} \cap \propositions{\psi_{3-i}} = \emptyset$ for $i \in \{1,2\}$. 
	Let $\neg(\varphi_1 \land \varphi_2 \land \varphi_3)$ be unrealizable. Then, $\varphi$ is realizable if, and only if, both $\varphi' = (\varphi_1 \land \varphi_3) \rightarrow \psi_1$ and $\varphi'' = (\varphi_2 \land \varphi_3) \rightarrow \psi_2$ are realizable.
\end{theorem}
%

\begin{proof}
    Define $V_i = \propositions{\varphi_i} \cup \propositions{\varphi_3} \cup \propositions{\psi_3}$ for $i \in \{1,2\}$. Since we have $V = \propositions{\varphi}$ by assumption, $V_1 \cup V_2 = V$ holds.
    With the assumptions made on $\varphi_1$, $\varphi_2$, $\varphi_3$, $\psi_1$, and $\psi_2$, we obtain $V_1 \cap V_2 \subseteq \inputs$.
    
    First, let $\varphi$ be realizable and let $f:(2^V)^* \times 2^\inputs \rightarrow 2^\outputs$ be an implementation that realizes $\varphi$. Let $\sigma \in \compatibleWords{f}$. Then, $\sigma \in \mathcal{L}(\varphi)$ and thus by the semantics of implication, $\restrict{\sigma}{(V \setminus \propositions{\psi_{3-i}})} \in \mathcal{L}((\varphi_1 \land \varphi_2 \land \varphi_3) \rightarrow \psi_i)$ follows for $i \in \{1,2\}$. Hence, an implementation $f_i$ that behaves as $f$ restricted to $\outputs \setminus \propositions{\psi_{3-i}}$ realizes $(\varphi_1 \land \varphi_2 \land \varphi_3) \rightarrow \psi_i$.
    By \Cref{lem:assumption_dropping}, $(\varphi_1 \land \varphi_2 \land \varphi_3) \rightarrow \psi_i$ and $(\varphi_i \land \varphi_3) \rightarrow \psi_i$ are equirealizable since $\varphi_1$, $\varphi_2$, and~$\varphi_3$ as well as $\varphi_{3-i}$ and $\psi_i$ do not share any variables. Thus, there exist implementations $f_1$ and $f_2$ realizing $(\varphi_1 \land \varphi_3) \rightarrow \psi_1$ and $(\varphi_2 \land \varphi_3) \rightarrow \psi_2$, respectively.
    
    Next, let both $(\varphi_1 \land \varphi_3) \rightarrow \psi_1$ and $(\varphi_2 \land \varphi_3) \rightarrow \psi_2$ be realizable and let $f_i: (2^{V_i})^* \times 2^{\inputs \cap V_i} \rightarrow 2^{\outputs \cap V_i}$ be an implementation realizing $(\varphi_i \land \varphi_3) \rightarrow \psi_i$. We construct an implementation $f:(2^V)^* \times 2^\inputs \rightarrow 2^\outputs$ from $f_1$ and $f_2$ as follows: $f(\sigma,\inp{i}) := f_1(\restrict{\sigma}{V_1},\inp{i} \cap V_1) \cup f_2(\restrict{\sigma}{V_2},\inp{i} \cap V_2)$. 
    Let $\sigma \in \compatibleWords{f}$.
    Since $V_1$ and $V_2$ do not share any output variables, $\restrict{\sigma}{V_i} \in \mathcal{L}((\varphi_i \land \varphi_3) \rightarrow \psi_i)$ follows from the construction of $f$. Moreover, $\restrict{\sigma}{V_1}$ and $\restrict{\sigma}{V_2}$ agree on shared variables and thus $(\restrict{\sigma}{V_1}) \cup (\restrict{\sigma}{V_2}) \in \mathcal{L}(\varphi' \land\varphi'')$ holds. Therefore, we have $(\restrict{\sigma}{V_1}) \cup (\restrict{\sigma}{V_2}) \in \mathcal{L}(\varphi)$ as well by the semantics of conjunction and implication. Since $V_1 \cup V_2 = V$, we have $(\restrict{\sigma}{V_1}) \cup (\restrict{\sigma}{V_2}) = \sigma$ and thus $\sigma \in \mathcal{L}(\sigma)$. Hence, $f$ realizes $\varphi$.\qed
\end{proof}

Analyzing assumptions thus allows for decomposing LTL formulas in further cases and still ensures soundness and completeness of modular synthesis.
In the following, we present an optimized LTL decomposition algorithm that incorporates assumption dropping into the search for independent conjuncts. 
Intuitively, the algorithm needs to identify variables that cannot be shared safely among subspecifications. If an \emph{assumption} contains such non-sharable variables, we say that it is \emph{bound} to guarantees since it can influence the possible decompositions. Otherwise, it is called \emph{free}.


%

To determine which assumptions are relevant for decomposition, \ie, which assumptions are \emph{bounded assumptions}, we build a slightly modified version of the dependency graph that is only based on assumptions and not on all conjuncts of the formula. Moreover, all variables serve as the nodes of the graph, not only the output variables. An undirected edge between two variables in the modified dependency graph denotes that the variables occur in the same assumption.
Variables that are contained in the same connected component as an output variable $o\in\outputs$ are thus connected to $o$ over a path of one or more assumptions. Therefore, they may not be shared among subspecifications as they might influence $o$ and thus may influence the decomposability of the specification.
These variables are then called \emph{decomposition-critical}.
Given the modified dependency graph, we can compute the decomposition-critical propositions with a simple depth-first search.

\begin{algorithm}[t]
    \DontPrintSemicolon
    \KwIn{$\varphi$: LTL, \var{inp}: List Variable, \var{out}: List Variable}
    \KwResult{\var{specs}: List (LTL, List Variable, List Variable)}
    \var{assumptions} $\leftarrow$ getAssumptions($\varphi$)\;
    \var{guarantees} $\leftarrow$ getGuarantees($\varphi$)\;
    \var{decCritProps} $\leftarrow$ getDecCritProps($\varphi$)\;
    \var{graph} $\leftarrow$ buildDependencyGraph($\varphi$,\var{decCritProps})\;
    \var{components} $\leftarrow$ \var{graph}.connectedComponents()\;
    \var{specs} $\leftarrow$ new LTL[$|$\var{components}$|+1$]\;
    \var{freeAssumptions} $\leftarrow$[\ ]\;
    \ForEach{\upshape{$\psi \in$ \var{assumptions}}}{
        \var{propositions} $\leftarrow$ \var{decCritProps} $\cap$ getProps($\psi$)\;
        \eIf{\upshape{$|$\var{propositions}$| = 0$}}{
            \var{freeAssumptions}.append($\psi$)\;
        }{
            \ForEach{\upshape{(\var{spec}, \var{set}) $\in$ zip \var{specs} (\var{components} $++$ [\var{inp}])}}{
                \If{\upshape{\var{propositions} $\cap$ \var{set} $\neq \emptyset$}}{
                    \var{spec}.addAssumption($\psi$)\;
                    break\;
                }
            }	
        }
    }
    \ForEach{\upshape{$\psi \in$ \var{guarantees}}}{
        \var{propositions} $\leftarrow$ \var{decCritProps} $\cap$ getProps($\psi$)\;
        \ForEach{\upshape{(\var{spec}, \var{set}) $\in$ zip \var{specs} (\var{components} $++$ [\var{inp}])}}{
            \If{\upshape{\var{propositions} $\cap$ \var{set} $\neq \emptyset$}}{
                \var{spec}.addGuarantee($\psi$)\;
                break\;
            }
        }
        
    }
    \KwRet{\upshape{addFreeAssumptions \var{specs freeAssumptions}}}

    \caption{Optimized LTL Decomposition Algorithm}
    \label{alg:optimized_decomposition}
\end{algorithm} 

After computing the decomposition-critical propositions, we create the dependency graph and extract connected components in the same way as in \Cref{alg:rewriting-based_decomposition} to decompose the LTL specification. Instead of using only output variables as nodes of the graph, though, we use all decomposition-critical variables.
We then exclude free assumptions and add all other assumptions to their respective subspecification similar to \Cref{alg:rewriting-based_decomposition}. We assign the guarantees to their subspecification in the same manner.
%
%
%
Lastly, we add the remaining assumptions. Since all of these assumptions are free, they could be safely added to all subspecifications. Yet, to obtain small subspecifications, we only add them to subspecifications for which they are needed. Note that we have to add all assumptions featuring an input variable that occurs in the subspecification. Therefore, we analyze the assumptions and add them in one step, as a naive approach could have an unfavorable running time. 
The whole LTL decomposition algorithm with optimized assumption handling is shown in \Cref{alg:optimized_decomposition}.

The decomposition algorithm does not check for assumption violations. 
The unrealizability of the negation of the dropped assumption, however, is an essential part of the criterion for assumption dropping (c.f.\ \Cref{thm:ltl_decomposition_with_assumptions}).
Therefore, we incorporate the check for assumption violations into the modular synthesis algorithm: Before decomposing the specification, we perform synthesis on the negated assumptions. If synthesis returns that the negated assumptions are realizable, the system is able to violate an assumption. The implementation satisfying the negated assumptions is then extended to an implementation for the whole specification that violates the assumptions and thus realizes the specification. Otherwise, if the negated assumptions are unrealizable, the conditions of \Cref{thm:ltl_decomposition_with_assumptions} are satisfied. Hence, we can use the decomposition algorithm and proceed as in \Cref{alg:compositional_synthesis}. The modified modular synthesis algorithm that incorporates the check for assumption violations is shown in \Cref{alg:compositional_synthesis2}.

\begin{algorithm}[t]
    \DontPrintSemicolon
    \KwIn{\var{s}: Specification, \var{inp}: List Variable, \var{out}: List Variable}
    \KwResult{\var{realizable}: Bool, \var{implementation}: $\mathcal{T}$}
    
    (\var{real}, \var{strat}) $\leftarrow$ synthesize(getNegAss($\varphi$), \var{inp}, \var{out})\;
    \If{\upshape{\var{real}}}{
        \KwRet{\upshape{($\top$, \var{strat})}}
    }
    
    \var{subspecifications} $\leftarrow$ decompose$(\var{s},\var{inp},\var{out})$ \;
    \var{sub\_results} $\leftarrow$ map synthesize \var{subspecifications} \;
    \ForEach{\upshape{(\var{real}, \var{strat}) $\in$ \var{sub\_results}}}{
        \If{\upshape{! \var{real}}}{
            \var{implementation} $\leftarrow$ extendCounterStrategy(\var{strat}, \var{s})\;
            \KwRet{\upshape{($\bot$, \var{implementation})}}
        }
    }
    \var{impls} $\leftarrow$ map second \var{sub\_results}\;
    \var{implementation} $\leftarrow$ compose \var{impls}\;
    \KwRet{\upshape{($\top$, \var{implementation})}}
    
    \caption{Modular Synthesis Algorithm with Optimized LTL Decomposition}
    \label{alg:compositional_synthesis2}
\end{algorithm}


Note that \Cref{alg:optimized_decomposition} is only applicable to specifications in a strict assume-guarantee format since \Cref{thm:ltl_decomposition_with_assumptions} assumes a top-level implication in the formula.
In the next section, we thus present an extension of the LTL decomposition algorithm with optimized assumption handling to specifications consisting of several assume-guarantee conjuncts, \ie, specifications of the form $\varphi = (\varphi_1 \rightarrow \psi_1) \land \dots \land (\varphi_k \rightarrow \psi_k)$.


\section{Optimized LTL Decomposition for Formulas with Several Assume-Guarantee Conjuncts}\label{sec:ltl_optimized}

Since \Cref{cor:equisynthesizeability_independent_conjuncts} can be applied recursively, classical LTL decomposition, \ie, as described in \Cref{alg:rewriting-based_decomposition}, is applicable to specifications with several conjuncts. That is, in particular, it is applicable to specifications with several assume-guarantee conjuncts, \ie, specifications of the form $\varphi = (\varphi_1 \rightarrow \psi_1) \land \dots \land (\varphi_k \rightarrow \psi_k)$.
\Cref{alg:optimized_decomposition}, in contrast, is restricted to LTL specifications consisting of a single assume-guarantee pair since \Cref{thm:ltl_decomposition_with_assumptions}, on which \Cref{alg:optimized_decomposition} relies, assumes a top-level implication in the specification.
Hence, we cannot apply the optimized assumption handling to specifications with several assume-guarantee conjuncts directly.

A naive approach to extend assumption dropping to formulas with several assume-guarantee conjuncts is to first drop assumptions for all conjuncts separately and then to decompose the resulting specification using \Cref{alg:rewriting-based_decomposition}.
In general, however, this is not sound: The other conjuncts may introduce dependencies between assumptions and guarantees that prevent the dropping of the assumption. When considering the conjuncts during the assumption dropping phase separately, however, such dependencies are not detected.
For instance, consider a system with $\inputs = \{i\}$, $\outputs = \{o_1,o_2\}$, and the specification $\varphi = \Globally\neg(o_1 \land o_2) \land \Globally \neg(i \leftrightarrow o_1) \land (\Globally i \rightarrow \Globally o_2)$.
Clearly, $\varphi$ is realizable by an implementation that sets $o_1$ to $\neg i$ and $o_2$ to $i$ in every time step.
Since the first conjunct contains both $o_1$ and $o_2$, \Cref{cor:equisynthesizeability_independent_conjuncts} is not applicable and thus \Cref{alg:rewriting-based_decomposition} does not decompose~$\varphi$.
The naive approach for incorporating assumption dropping described above considers the third conjunct of $\varphi$ separately and checks whether whether the assumption $\Globally i$ can be dropped. Since the assumptions and guarantees do not share any variables, \Cref{lem:assumption_dropping} is applicable and thus the naive algorithm drops $\Globally i$, yielding $\varphi' = \Globally\neg(o_1 \land o_2) \land \Globally \neg(i \leftrightarrow o_1) \land \Globally o_2$. Yet, $\varphi'$ is not realizable: If $i$ is constantly set to $\false$, the second conjunct of $\varphi'$ enforces $o_1$ to be always set to $\true$. The third conjunct enforces that $o_2$ is constantly set to $\true$ irrespective of the input $i$. The first conjunct, however, requires in every time step one of the output variables to be $\false$.
Thus, although \Cref{lem:assumption_dropping} is applicable to $\Globally i \rightarrow \Globally o_1$, dropping the assumption safely is not possible in the context of the other two conjuncts.
In particular, the first conjunct of $\varphi$ introduces a dependency between $o_1$ and $o_2$ while the second conjunct introduces one between $i$ and $o_1$. Hence, there is a transitive dependency between $i$ and $o_1$ due to which the assumption $\Globally i$ cannot be dropped. This dependency is not detected when considering the conjuncts separately during the assumption dropping phase.

In this section, we introduce an optimization of the LTL decomposition algorithm which is able to decompose specifications with several conjuncts (possibly) in assume-guarantee format and which is, in contrast to the naive approach described before, sound.
Similar to the naive approach, the main idea is to first check for assumptions that can be dropped in the different conjuncts and to then perform the classical LTL decomposition algorithm. Yet, the assumption dropping phase is not performed completely separately for the individual conjuncts but takes the other conjuncts and thus possible transitive dependencies between the assumptions and guarantees into account.

If the other conjuncts do not share any variable with the assumption to be dropped, then there are no transitive dependencies between the assumption and the guarantee due to the other conjuncts. Thus, the assumption can be dropped safely if the other conditions of \Cref{lem:assumption_dropping} are satisfied:

\begin{lemma}\label{lem:assumption_dropping_optimized}
    Let $\varphi = \psi_1 \land ((\varphi_1 \land \varphi_2) \rightarrow \psi_2)$ be an LTL~formula, where we have $\propositions{\varphi_1} \cap \propositions{\varphi_2} = \emptyset$, $\propositions{\varphi_2} \cap \propositions{\psi_1} = \emptyset$ and $\propositions{\varphi_2} \cap \propositions{\psi_2} = \emptyset$. Let~$\neg\varphi_2$ be unrealizable.
    Then, $\varphi' = \psi_1 \land (\varphi_1 \rightarrow \psi_2)$ is realizable if, and only if, $\varphi$ is realizable.
\end{lemma}
\begin{proof}
	Let $V_1 := \propositions{\psi_1 \land (\varphi_1 \rightarrow \psi_2)}$, $\inputs_1 := I \cap V_1$, and $\outputs_1 := O \cap V_1$.
	If $\varphi'$ is realizable, then we can construct an implementation $f: (2^V)^* \times 2^\inputs \rightarrow 2^\outputs$ that realizes $\varphi$ from the implementation $f_1: (2^{V_1})^* \times 2^{\inputs_1} \rightarrow 2^{\outputs_1}$ that realizes $\varphi'$ analogous to the proof of \Cref{lem:assumption_dropping}.
	
	If $\varphi$ is realizable, then there is an implementation $f: (2^V)^* \times 2^\inputs \rightarrow 2^\outputs$ that realizes $\varphi$.
	Since $\neg\varphi_2$ is unrealizable by assumption, there is a counterstrategy $f^c_2: (2^{\propositions{\varphi_2}})^* \rightarrow 2^{\inputs \cap \propositions{\varphi_2}}$ for $\neg \varphi_2$ and all words compatible with $f^c_2$ satisfy $\varphi_2$.
	Let $g: (2^{V_1})^* \times 2^{\inputs_1} \rightarrow 2^{\outputs_1}$ be the implementation constructed from $f$ and $f^c_2$ in the proof of \Cref{lem:assumption_dropping}. We show that $g$ realizes $\varphi'$.
	Let $\sigma \in \compatibleWords{g}$ and let $\sigma_f$ be the corresponding infinite sequence obtained from $g$ when not restricting the output of $f$ to the variables in $O_1$.
	As shown in the proof of \Cref{lem:assumption_dropping}, $\sigma_\mathit{f} \in \mathcal{L}(\varphi)$ and $\sigma_\mathit{f} \in \mathcal{L}(\varphi_2)$.
	Thus, $\sigma_\mathit{f} \in \mathcal{L}(\psi_1 \land (\varphi_1 \rightarrow \psi_2))$.
    Since $\varphi_2$ neither shares variables with $\varphi_1$ nor with $\psi_1$ or $\psi_2$, the satisfaction of $\psi_1 \land (\varphi_1 \rightarrow \psi_2)$ is not influenced by the variables outside of $V_1$. 
    Hence, since $\sigma_\mathit{f} \cap V_1 = \sigma$ by construction, $\sigma \in \mathcal{L}(\varphi')$ follows and thus $g$ realizes $\varphi'$.\qed	
\end{proof}

Similar to the optimized assumption handling for specifications in strict assume-guarantee form described in the previous section, we utilize \Cref{lem:assumption_dropping_optimized} for an optimized decomposition for specifications containing several assume-guarantee conjuncts:
We rewrite LTL formulas of the form $\varphi = \psi' \land \bigwedge^m_{i=1} \varphi_i \rightarrow \bigwedge^n_{j=1} \psi_j$ to $\psi' \land \bigwedge^n_{j=1}(\bigwedge^m_{i=1} \varphi_i \rightarrow \psi_j)$ and then drop assumptions for the individual guarantees $\psi_1, \dots, \psi_j$ according to \Cref{lem:assumption_dropping_optimized}. If the resulting subspecifications only share input variables, they are equirealizable to $\varphi$.

\begin{theorem}\label{thm:ltl_decomposition_with_assumptions_optimized}
	$\!\!\!$ Let $\varphi \! = \! \psi'_1 \land \psi'_2 \land (\!(\varphi_1 \land \varphi_2 \land \varphi_3) \!\rightarrow \!\psi_1 \land \psi_2)$ be an LTL formula over $V$, where
	$\propositions{\varphi_3} \subseteq \inputs$ and
	$(\propositions{\psi_1} \cup \propositions{\psi'_1}) \cap (\propositions{\psi_2} \cup \propositions{\psi'_2})\subseteq \inputs$. 
	Let $\propositions{\varphi_i} \cap \propositions{\varphi_j} = \emptyset$ for $i,j \in \{1,2,3\}$ with $i \neq j$, and let
	$\propositions{\varphi_i} \cap \propositions{\psi_{3-i}} = \emptyset$ for $i \in \{1,2\}$.
	Let $\propositions{\psi'_i} \cap \propositions{\varphi_{3-i}} = \emptyset$ for $i \in \{1,2\}$.
	Moreover, let $\neg(\varphi_1 \land \varphi_2 \land \varphi_3)$ be unrealizable. Then, $\varphi$ is realizable if, and only if, both $\varphi' = \psi' \land ((\varphi_1 \land \varphi_3) \rightarrow \psi_1)$ and $\varphi'' = \psi'' \land ((\varphi_2 \land \varphi_3) \rightarrow \psi_2)$ are realizable.
\end{theorem}
\begin{proof}
	First, let $\varphi'$ and $\varphi''$ be realizable. Then, there are implementations $f_1$ and $f_2$ realizing $\varphi'$ and~$\varphi''$, respectively.
	Since $\varphi'$ and $\varphi''$ do not share output variables by assumption, we can construct an implementation realizing $\varphi$ from $f_1$ and $f_2$ as in the proof of \Cref{thm:ltl_decomposition_with_assumptions}.
	
	Next, let $\varphi$ be realizable and let $f: (2^V)^* \times 2^\inputs \rightarrow 2^\outputs$ be an implementation realizing $\varphi$. Let $\sigma \in \compatibleWords{f}$. Then, $\sigma \in \mathcal{L}(\varphi)$ holds. Let $V' = \propositions{\varphi'} \cup \propositions{\varphi_2}$ and let $V'' = \propositions{\varphi''} \cup \propositions{\varphi_1}$.
	Then, since $\sigma \in \mathcal{L}(\varphi)$ holds, $\sigma \cap V' \in \mathcal{L}(\psi' \land ((\varphi_1 \land \varphi_2 \land \varphi_3) \rightarrow \psi_1))$ as well as $\sigma \cap V'' \in \mathcal{L}(\psi'' \land ((\varphi_1 \land \varphi_2 \land \varphi_3) \rightarrow \psi_2))$ follow.
	Thus, an implementation $f_1$ that behaves as $f$ restricted to the variables in $V'$ realizes $\psi' \land ((\varphi_1 \land \varphi_2 \land \varphi_3) \rightarrow \psi_1)$. An  implementation $f_2$ that behaves as $f$ restricted to the variables in $V''$ realizes $\psi'' \land ((\varphi_1 \land \varphi_2 \land \varphi_3) \rightarrow \psi_2)$.
	By assumption, for $i \in \{1,2\}$, $\varphi_{i}$ does not share any variables with $\varphi_3$, $\varphi_{3-1}$, $\psi_{3-1}$ and $\psi'_{3-1}$. 
	Therefore, by \Cref{lem:assumption_dropping_optimized}, $\psi'_1 \land ((\varphi_1 \land \varphi_2 \land \varphi_3) \rightarrow \psi_1)$ and $\varphi'$ are equirealizable. Moreover, $\psi'_2 \land ((\varphi_1 \land \varphi_2 \land \varphi_3) \rightarrow \psi_2)$ and $\varphi''$ are equirealizable.
	Thus, since $f_1$ and $f_2$ realize the former formulas, $\varphi'$ and $\varphi''$ are both realizable. \qed
\end{proof}

\begin{algorithm}[t]
    \DontPrintSemicolon
    \KwIn{$\varphi$: LTL, \var{inp}: List Variable, \var{out}: List Variable}
    \KwResult{\var{specs}: List (LTL, List Variable, List Variable)}
    \var{implication} $\leftarrow$ chooseImplication($\varphi$)\;
    \var{assumptions} $\leftarrow$ getAssumptions(\var{implication})\;
    \var{guarantees} $\leftarrow$ getGuarantees(\var{implication})\;
    \var{decCritProps} $\leftarrow$ getDecCritProps(\var{implication})\;
    \var{graph} $\leftarrow$ buildDependencyGraph($\varphi$,\var{decCritProps})\;
    \var{components} $\leftarrow$ \var{graph}.connectedComponents()\;
    \var{specs} $\leftarrow$ new LTL[$|$\var{components}$|+1$]\;
    \var{freeAssumptions} $\leftarrow$[\ ]\;
    \ForEach{\upshape{$\psi \in$ \var{assumptions}}}{
        \var{propositions} $\leftarrow$ \var{decCritProps} $\cap$ getProps($\psi$)\;
        \eIf{\upshape{$|$\var{propositions}$| = 0$}}{
            \var{freeAssumptions}.append($\psi$)\;
        }{
            \ForEach{\upshape{(\var{spec}, \var{set}) $\in$ zip \var{specs} (\var{components} $++$ [\var{inp}])}}{
                \If{\upshape{\var{propositions} $\cap$ \var{set} $\neq \emptyset$}}{
                    \var{spec}.addAssumption($\psi$)\;
                    break\;
                }
            }	
        }
    }
    \ForEach{\upshape{$\psi \in$ \var{guarantees}}}{
        \var{propositions} $\leftarrow$ \var{decCritProps} $\cap$ getProps($\psi$)\;
        \ForEach{\upshape{(\var{spec}, \var{set}) $\in$ zip \var{specs} (\var{components} $++$ [\var{inp}])}}{
            \If{\upshape{\var{propositions} $\cap$ \var{set} $\neq \emptyset$}}{
                \var{spec}.addGuarantee($\psi$)\;
                break\;
            }
        }
        
    }
    \ForEach{\upshape{$\psi \in$ getConjuncts($\varphi$)$\setminus$\var{implication}}}{
        \var{propositions} $\leftarrow$ \var{decCritProps} $\cap$ getProps($\psi$)\;
        \ForEach{\upshape{(\var{spec}, \var{set}) $\in$ zip \var{specs} (\var{components} $++$ [\var{inp}])}}{
            \If{\upshape{\var{propositions} $\cap$ \var{set} $\neq \emptyset$}}{
                \var{spec}.addConjunct($\psi$)\;
                break\;
            }
        }
        
    }
    \KwRet{\upshape{addFreeAssumptions \var{specs freeAssumptions}}}

    \caption{Optimized LTL Decomposition Algorithm for Specifications with Conjuncts}
    \label{alg:optimized_decomposition_with_conjuncts}
\end{algorithm} 

Utilizing \Cref{thm:ltl_decomposition_with_assumptions_optimized}, we extend \Cref{alg:optimized_decomposition} to LTL specifications that do not follow a strict assume-guarantee form but consist of multiple conjuncts. The extended algorithm is depicted in \Cref{alg:optimized_decomposition_with_conjuncts}. We assume that the specification is not decomposable by \Cref{alg:rewriting-based_decomposition}, \ie, we assume that no plain decompositions are possible.
In practice, we thus first rewrite the specification and apply \Cref{alg:rewriting-based_decomposition} afterwards before then applying \Cref{alg:optimized_decomposition_with_conjuncts} to the resulting subspecifications.

Hence, we assume that the dependency graph built from the output propositions of all given conjuncts consists of a single connected component. \Cref{thm:ltl_decomposition_with_assumptions_optimized} hands us the tools to ``break a link'' in that chain of dependencies. This link has to be induced by a suitable implication.
\Cref{alg:optimized_decomposition_with_conjuncts} assumes that at least one of the conjuncts is an implication. In case of more than one implication, the choice of the implication consequently determines whether or not a decomposition os found. Therefore, it is crucial to reapply the algorithm on the subspecifications after a decomposition has been found and to try all implications if no decomposition is found. Since iterating through all conjuncts does not pose a large overhead in computing time, the choice of the implication is not further specified in the algorithm.

The extended algorithm is similar to \Cref{alg:optimized_decomposition}.
Note that the dependency graph used for finding the decomposition-critical propositions is built only from the assumptions of the chosen implication as we are only seeking for droppable assumptions of this implication.
In contrast to \Cref{alg:optimized_decomposition}, the dependency graph in line 5 of \Cref{alg:optimized_decomposition_with_conjuncts} also includes the dependencies induced by the other conjuncts, similarly to the dependency graph in \Cref{alg:rewriting-based_decomposition}. Here, we consider all decomposition-critical variables in the conjuncts, not only output variables, as an assumption can only be dropped if there are no shared variables with the remaining conjuncts. Therefore, the additional conjuncts are treated in the same way as the guarantees.
This carries over to when the conjuncts are added to the subspecifications.
Lastly, \Cref{alg:optimized_decomposition_with_conjuncts} slightly differs from \Cref{alg:optimized_decomposition} when the free assumptions are added to the subspecifications. Here, the  remaining conjuncts have to be considered, too, since we may not drop assumptions that share variables with the outside conjunct. Consequently, all free assumptions that share an input with one of the remaining conjuncts, needs to be added.

One detail that has to be taken into account when integrating this LTL decomposition algorithm with extended optimized assumption handling into a synthesis tool, is that, like \Cref{alg:optimized_decomposition}, \Cref{alg:optimized_decomposition_with_conjuncts} assumes that all negated assumptions are unrealizable. For formulas in a strict assume-guarantee format, the consequences of realizable assumptions is that we have found a strategy for the implementation. This changes when considering formulas with additional conjuncts since they might forbid this strategy. To detect such strategies, we can verify the synthesized strategy against the remaining conjunct and only extend it to a counterstrategy for the whole specification in the positive case.


\section{Experimental Evaluation}

We implemented the modular synthesis algorithm as well as the decomposition approaches and evaluated them on the 346 publicly available SYNTCOMP~\cite{SYNTCOMP} 2020 benchmarks. Note that only 207 of the benchmarks have more than one output variable and are therefore realistic candidates for decomposition.
The automaton decomposition algorithm utilizes Spot's~\cite{Duret-LutzLFMRX16} automaton library (Version 2.9.6). The LTL decomposition relies on SyFCo~\cite{JacobsFS16} for formula transformations (Version 1.2.1.1).
We first decompose the specification with our algorithms and then run synthesis on the resulting subspecifications.
We compare the CPU time of the synthesis task as well as the number of gates, and latches of the synthesized AIGER circuit for the original specification to the sum of the corresponding attributes of all subspecifications.
Thus, we calculate the runtime for sequential modular synthesis. 
Parallelization of the synthesis tasks may further reduce the runtime.

\subsection{LTL Decomposition}

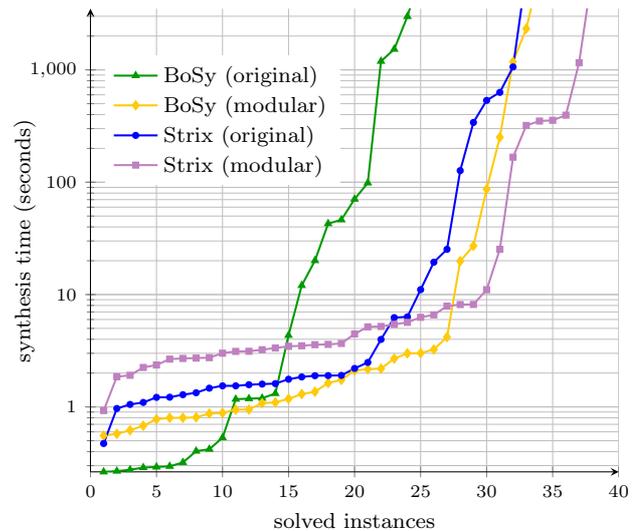
\begin{figure}[t]
    \centering
    \scalebox{0.96}{
    \begin{tikzpicture}
    \begin{axis}[
    width=1.05\linewidth,
    height=0.95\linewidth,
    axis lines=left,
    grid=both,
    legend style={draw=none, font=\footnotesize},
    legend cell align=left,
    legend style={at={(0.25,0.9)}, anchor=north},
    label style={font=\footnotesize},
    ticklabel style={font=\scriptsize},
    xmin=0,
    ymax=3500,
    xmax=40,
    xtick={0,5,...,40},
    log ticks with fixed point,
    ymode=log,
    xlabel=solved instances,
    ylabel=synthesis time (seconds),
    xlabel style={xshift=0cm, yshift=0cm},
    ylabel style={xshift=0cm, yshift=-0.3cm},
    legend entries={{BoSy (original)}, {BoSy (modular)}, {Strix (original)}, {Strix (modular)}},
    ]
    \addplot[color=black!40!green, mark=triangle*, mark options={scale=0.6}, thick] coordinates {
        ( 1,0.264)
        ( 2,0.268)
        ( 3,0.276)
        ( 4,0.288)
        ( 5,0.292)
        ( 6,0.296)
        ( 7,0.32)
        ( 8,0.404)
        ( 9,0.42)
        ( 10,0.532)
        ( 11,1.172)
        ( 12,1.184)
        ( 13,1.196)
        ( 14,1.316)
        ( 15,4.336)
        ( 16,12.048)
        ( 17,20.06)
        ( 18,42.804)
        ( 19,46.288)
        ( 20,70.71)
        ( 21,98.56)
        ( 22,1192.84)
        ( 23,1526.32)
        ( 24,2986.78)
        ( 25,7200) };
    \addplot[color=yellow!80!red, mark=diamond*, mark options={scale=0.75}, thick] coordinates {
        ( 1,0.556)
        ( 2,0.576)
        ( 3,0.62)
        ( 4,0.68)
        ( 5,0.776)
        ( 6,0.8)
        ( 7,0.8)
        ( 8,0.808)
        ( 9,0.872)
        ( 10,0.884)
        ( 11,0.944)
        ( 12,0.948)
        ( 13,1.08)
        ( 14,1.096)
        ( 15,1.18)
        ( 16,1.3)
        ( 17,1.364)
        ( 18,1.632)
        ( 19,1.74)
        ( 20,2.108)
        ( 21,2.168)
        ( 22,2.192)
        ( 23,2.692)
        ( 24,2.996)
        ( 25,3.008)
        ( 26,3.244)
        ( 27,4.188)
        ( 28,19.846)
        ( 29,27.136)
        ( 30,86.68)
        ( 31,251.288)
        ( 32,1179.68)
        ( 33,2326.39)
        ( 34,7200) };
    \addplot[color=blue, mark=*,  mark options={scale=0.5}, thick] coordinates {
        (1, 0.472) (2, 0.968) (3, 1.052) (4, 1.096)
        (5, 1.216)
        (6, 1.22)
        (7, 1.284)
        (8, 1.336)
        (9, 1.468)
        (10, 1.54)
        (11, 1.54)
        (12, 1.576)
        (13, 1.596)
        (14, 1.612)
        (15, 1.764)
        (16, 1.848)
        (17, 1.896)
        (18, 1.9)
        (19, 1.908)
        (20, 2.2)
        (21, 2.488)
        (22, 3.988)
        (23, 6.232)
        (24, 6.336)
        (25, 11.06)
        (26, 19.408)
        (27, 25.28)
        (28, 126.808)
        (29, 340.34)
        (30, 534.732)
        (31, 630.492)
        (32, 1062.27) (33, 7200)};
    \addplot[color=violet!50!white, mark=square*, mark options={scale=0.5}, thick] coordinates {
        ( 1,0.928)
        ( 2,1.856)
        ( 3,1.912)
        ( 4,2.244)
        ( 5,2.36)
        ( 6,2.664)
        ( 7,2.7)
        ( 8,2.72)
        ( 9,2.744)
        ( 10,3.008)
        ( 11,3.12)
        ( 12,3.128)
        ( 13,3.228)
        ( 14,3.344)
        ( 15,3.464)
        ( 16,3.5)
        ( 17,3.58)
        ( 18,3.6)
        ( 19,3.672)
        ( 20,4.464)
        ( 21,5.156)
        ( 22,5.2)
        ( 23,5.428)
        ( 24,5.656)
        ( 25,6.272)
        ( 26,6.596)
        ( 27,7.892)
        ( 28,8.148)
        ( 29,8.168)
        ( 30,11.04)
        ( 31,25.292)
        ( 32,167.004)
        ( 33,319.988)
        ( 34,350.404)
        ( 35,355.62)
        ( 36,395.236)
        ( 37,1156.68)
        ( 38,7200)
    };

    \end{axis}
    
    \end{tikzpicture}}
	\caption{Comparison of the performance of modular and non-compositional synthesis with BoSy and Strix on the decomposable SYNTCOMP benchmarks. For the modular approach, the accumulated time for all synthesis tasks is depicted.}
    \label{comparison_ltl}
\end{figure}

\begin{table*}[t]
        \caption{Distribution of the number of subspecifications over all specifications for LTL decomposition.}
    \label{table:ltl_subspecs}
    \centering
    \begin{tabular}{p{2.9cm}|>{\centering}p{0.5cm}|>{\centering}p{0.5cm}|>{\centering}p{0.5cm}|>{\centering}p{0.5cm}|>{\centering}p{0.5cm}|>{\centering}p{0.5cm}|>{\centering}p{0.5cm}|>{\centering}p{0.5cm}|>{\centering}p{0.5cm}|>{\centering}p{0.5cm}|>{\centering}p{0.5cm}|>{\centering\arraybackslash}p{0.5cm}}
        \# subspecifications& 1 & 2 & 3 & 4 & 5 & 6 & 7 & 8 & 9 & 10 & 11 & 12 \\ 
        \hline 
        \# specifications& 308 & 19 & 8 & 2 & 3 & 2 & 0 & 2 & 0 & 1  & 1 & 1  \\ 
    \end{tabular} 
\end{table*}

\begin{table*}[t]
    \centering
    \caption{Synthesis time in seconds of BoSy and Strix for non-compositional and modular synthesis on exemplary SYNTCOMP benchmarks with a timeout of 60 minutes.}
    \label{ltl_times}
    \begin{tabular}{p{2.9cm}|>{\centering}p{1.6cm}|>{\centering}p{1.6cm}|>{\centering}p{1.6cm}|>{\centering}p{1.6cm}|>{\centering\arraybackslash}p{1.8cm}}
    	 & \multicolumn{2}{c|}{original} & \multicolumn{2}{c|}{modular} & ~\\
	     Benchmark & BoSy & Strix & BoSy & Strix & \# subspec.\\
	     \hline
	     Cockpitboard & 1526.32 & 11.06 & \textbf{2.108} & 8.168 & 8\\
		 Gamelogic & TO & 1062.27 & TO & \textbf{25.292} & 4\\
		 LedMatrix & TO & TO& TO & \textbf{1156.68} & 3\\
		 Radarboard & TO & 126.808 & \textbf{3.008}& 11.04 & 11\\		 
		 Zoo10 & 1.316 & 1.54 & \textbf{0.884}& 2.744 & 2\\
		 generalized\_buffer\_2 & 70.71 & 534.732 & \textbf{4.188}& 7.892 & 2\\
		 generalized\_buffer\_3 & TO & TO& \textbf{27.136}& 319.988 & 3\\
		 shift\_8 & \textbf{0.404}& 1.336 & 2.168& 3.6 & 8\\
		 shift\_10 & \textbf{1.172}& 1.896 & 2.692 & 4.464 & 10\\
		 shift\_12 & 4.336& 6.232 & \textbf{3.244} & 5.428 & 12
    \end{tabular}
\end{table*}

LTL decomposition with optimized assumption handling (c.f.\ \Cref{sec:ltl_optimized}) terminates on all benchmarks in less than 26ms. Thus, even for non-decomposable specifications, the overhead of trying to perform decompositions is negligible. The algorithm decomposes 39 formulas into several subspecifications, most of them yielding two or three subspecifications.
Only a handful of formulas are decomposed into more than six subspecifications.
The full distribution of the number of subspecifications for all specifications is shown in \Cref{table:ltl_subspecs}

We evaluate our modular synthesis approach with two state-of-the-art synthesis tools: BoSy~\cite{BoSy}, a bounded synthesis tool, and Strix~\cite{MeyerStrix}, a game-based synthesis tool, both in their 2019 release. We used a machine with a 3.6GHz quad-core Intel Xeon processor and 32GB RAM as well as a timeout of 60 minutes.

\begin{table}[t]
    \centering
    \caption{Gates of the synthesized solutions of BoSy and Strix for non-compositional and modular synthesis on exemplary SYNTCOMP benchmarks. Entry -- denotes that no solution was found within 60 minutes.}
    \label{ltl_gat}
\begin{tabular}{p{2.8cm}|>{\centering}p{0.8cm}|>{\centering}p{0.8cm}|>{\centering}p{0.8cm}|>{\centering\arraybackslash}p{0.8cm}}
    & \multicolumn{2}{c|}{original} & \multicolumn{2}{c}{modular}\\
	Benchmark & BoSy & Strix & BoSy & Strix \\
    \hline
    Cockpitboard            & 11 & \textbf{7} & 25 & 10 \\
    Gamelogic               & -- & 26 & -- & \textbf{21} \\
    LedMatrix               & -- & -- & -- & \textbf{97} \\
	Radarboard              & -- & \textbf{6} & 19 & \textbf{6} \\
    Zoo10                   & 14 & 15 & 15 & \textbf{13}  \\
    generalized\_buffer\_2  & \textbf{3} & 12 & \textbf{3} & 11\\
    generalized\_buffer\_3  & -- & -- & \textbf{20} & 3772 \\
    shift\_8                & 8 & \textbf{0} & 8 & 7  \\
    shift\_10               & 10 & \textbf{0} & 10 & 9  \\
    shift\_12               & 12 & \textbf{0} & 12 & 11 
\end{tabular}
\end{table}
\begin{table}[t]
    \centering
    \caption{Latches of the synthesixed solutions of BoSy and Strix for non-compositional and modular synthesis on exemplary SYNTCOMP benchmarks. Entry -- denotes that no solution was found within 60 minutes.}
    \label{ltl_lat}
\begin{tabular}{p{2.8cm}|>{\centering}p{0.8cm}|>{\centering}p{0.8cm}|>{\centering}p{0.8cm}|>{\centering\arraybackslash}p{0.8cm}}
    & \multicolumn{2}{c|}{original} & \multicolumn{2}{c}{modular}\\
	Benchmark & BoSy & Strix & BoSy & Strix \\
    \hline
    Cockpitboard            & 1 & \textbf{0} & 8 & \textbf{0} \\
    Gamelogic               & -- & \textbf{2} & -- & \textbf{2} \\
    LedMatrix               & -- & -- & -- & \textbf{5} \\
	Radarboard              & -- & \textbf{0} & 11 & \textbf{0} \\
    Zoo10                   & \textbf{1} & 2 & 2 & 2 \\
    generalized\_buffer\_2  & 69 & 47134 & \textbf{14} & 557 \\
    generalized\_buffer\_3  & -- & -- & \textbf{3} & 14 \\
    shift\_8                & 1 & \textbf{0} & 8 & \textbf{0} \\
    shift\_10               & 1 & \textbf{0} & 10 & \textbf{0} \\
    shift\_12               & 1 & \textbf{0} & 12 & \textbf{0}
\end{tabular}
\end{table}

\begin{table*}[t]
    \centering
    \caption{Distribution of the number of subspecifications over all specifications for NBA decomposition. For 79 specifications, the timeout (60min) was reached. For 39 specification, the memory limit (16GB) was reached.}
    \label{table:nba_subspecs}
    \begin{tabular}{p{1.5cm}|>{\centering}p{0.5cm}|>{\centering}p{0.5cm}|>{\centering}p{0.5cm}|>{\centering}p{0.5cm}|>{\centering}p{0.5cm}|>{\centering}p{0.5cm}|>{\centering}p{0.5cm}|>{\centering}p{0.5cm}|>{\centering}p{0.5cm}|>{\centering}p{0.5cm}|>{\centering}p{0.5cm}|>{\centering}p{0.5cm}|>{\centering}p{0.5cm}|>{\centering}p{0.5cm}|>{\centering}p{0.5cm}|>{\centering\arraybackslash}p{0.5cm}}

        \# subspec. & 1 & 2 & 3 & 4 & 5 & 6 & 7 & 8 & 9 & 10 & 12 & 14 & 19 & 20 & 24 & 36\\ 
        \hline 
        \# spec.& 192 & 9 & 8 & 6 & 2 & 3 & 1 & 2 & 1 & 1 & 4 & 1 & 1 & 2 & 1 & 1\\ 
    \end{tabular} 
\end{table*}

In \Cref{comparison_ltl}, the comparison of the accumulated runtimes of the synthesis tasks of the subspecifications and of the original formula is shown for the decomposable SYNTCOMP benchmarks.
For both BoSy and Strix, decomposition generates a slight overhead for small specifications. For larger and more complex specifications, however, modular synthesis decreases the execution time significantly, often by an order of magnitude or more.
Note that due to the negligible runtime of specification decomposition, the plot looks similar when considering all SYNTCOMP benchmarks.

\Cref{ltl_times} shows the running times of BoSy and Strix for modular and non-compositional synthesis on exemplary benchmarks. For modular synthesis, the accumulated running time of all synthesis tasks is depicted.
On almost all of them, both tools decrease their synthesis times with modular synthesis notably compared to the original non-compositional approaches. 
Particularly noteworthy is the benchmark \emph{generalized\_buffer\_3}. In the last synthesis competition, SYNTCOMP~2021, no tool was able to synthesize a solution for it within one hour. With modular synthesis, however, BoSy yields a result in less than 28 seconds.

In \Cref{ltl_gat,ltl_lat}, the number of gates and latches, respectively, of the AIGER circuits~\cite{BiereHW11} corresponding to the implementations computed by BoSy and Strix for modular and non-compositional synthesis are depicted for exemplary benchmarks.
For most specifications, the solutions of modular synthesis are of the same size or smaller in terms of gates than the solutions for the original specification.
The size of the solutions in terms of latches, however, varies. Note that BoSy does not generate solutions with less than one latch in general. Hence, the modular solution will always have at least as many latches as subspecifications.

\subsection{Automaton Decomposition}

Besides LTL specifications, Strix also accepts specifications given as deterministic parity automata (DPAs) in extended HOA format~\cite{Perez2019}, an automaton format well-suited for synthesis. Thus, our implementation for decomposing specifications given as NBAs performs \Cref{alg:automaton-based_decomposition}, converts the resulting automata to DPAs and then synthesizes solutions with Strix.

For 235 out of the 346 benchmarks, NBA decomposition terminates within ten minutes, yielding several subspecifications or proving that the specification is not decomposable. 
In 79 of the other cases, the tool timed out after 60 minutes and in the remaining 32 cases it reached the memory limit of 16GB or the internal limits of Spot. 
Note, however, that for 81 specifications even plain DPA generation failed.
The distribution of the number of subspecifications for all specifications is shown in \Cref{table:nba_subspecs}.
Thus, while automaton decomposition yields more fine-grained decompositions than the approximate LTL approach, it becomes infeasible when the specifications grow. Hence, the advantage of smaller synthesis subtasks cannot pay off.
However, the coarser LTL decomposition suffices to reduce the synthesis time on common benchmarks significantly. Thus, LTL decomposition is in the right balance between small subtasks and a scalable decomposition.

For 43 specifications, the automaton approach yields decompositions and many of them consist of four or more subspecifications. For 22 of these specifications, the LTL approach yields a decomposition as well.
Yet, they differ in most cases, as the automaton approach yields more fine-grained decompositions.

Recall that only 207 SYNTCOMP benchmarks are realistic candidates for decomposition. The automaton approach proves that 90 of those specifications (43.6\%) are not decomposable. Thus, our implementations yield decompositions for 33.33\% (LTL) and 36.75\% (NBA) of the potentially decomposable specifications. 
We observed that decomposition works exceptionally well for specifications that stem from real system designs, for instance the Syntroids~\cite{GeierH0F19} case study, indicating that modular synthesis is particularly beneficial in practice.

\section{Conclusion}

We have presented a modular synthesis algorithm that applies compositional techniques to reactive synthesis. It reduces the complexity of synthesis by decomposing the specification in a preprocessing step and then performing independent synthesis tasks for the subspecifications.
We have introduced a criterion for decomposition algorithms that ensures soundness and completeness of modular synthesis as well as two algorithms for specification decomposition satisfying the criterion: A semantically precise one for specifications given as nondeterministic Büchi automata, and an approximate algorithm for LTL specifications.
We presented optimizations of the LTL decomposition algorithm for formulas in a strict assume-guarantee format and for formulas consisting of several assume-guarantee conjuncts. Both optimizations are based on dropping assumptions that do not influence the realizability of the rest of the formula.
We have implemented the modular synthesis algorithm as well as both decomposition algorithms and we compared our approach for the state-of-the-art synthesis tools BoSy and Strix to their non-compositional forms.
Our experiments clearly demonstrate the significant advantage of modular synthesis with LTL decomposition over traditional synthesis algorithms. While the overhead is negligible, both BoSy and Strix are able to synthesize solutions for more benchmarks with modular synthesis than in their non-compositional form. 
Moreover, on large and complex specifications, BoSy and Strix improve their synthesis times notably, demonstrating that specification decomposition is a game-changer for practical LTL synthesis.

Building up on the presented approach, we can additionally analyze whether the subspecifications fall into fragments for which efficient synthesis algorithms exist, for instance safety specifications. Since modular synthesis performs independent synthesis tasks for the subspecifications, we can choose, for each synthesis task, an algorithm that is tailored to the fragment the respective subspecification lies in.
 Moreover, parallelizing the individual synthesis tasks may increase the advantage of modular synthesis over classical algorithms.
Since the number of subspecifications computed by the LTL decomposition algorithm highly depends on the rewriting of the initial formula, a further promising next step is to develop more sophisticated rewriting algorithms.

\bibliographystyle{spmpsci}      
\bibliography{bib}   


\end{document}